\documentclass[letterpaper,11pt]{article}

\usepackage[margin=1in]{geometry}
\usepackage{amsfonts,amsmath,amssymb,amsthm,graphicx}
\usepackage{mathtools}
\usepackage{mathrsfs}
\usepackage{ifthen}
\usepackage{changepage}
\usepackage{enumerate}
\usepackage{scalerel}
\usepackage{accents}
\usepackage{enumitem}
\usepackage{subcaption} 
\usepackage{float}      

\usepackage[utf8]{inputenc} 
\usepackage[T1]{fontenc} 
\usepackage{csquotes}
\usepackage{comment}
\usepackage{bbm}
\usepackage{array}
\usepackage{makecell}
\usepackage{booktabs}
\usepackage{xfrac}

\usepackage[suppress]{color-edits}  
\addauthor{wt}{purple}
\addauthor{ai}{red}

\usepackage{tcolorbox}
\DeclareRobustCommand{\mybox}[2][gray!15]{
\begin{tcolorbox}[  
        left=0pt,
        right=0pt,
        top=0pt,
        bottom=0pt,
        colback=#1,
        colframe=#1,
        enlarge left by=0mm,
        boxsep=10pt,
        arc=2pt,outer arc=2pt,
        ]
        #2
\end{tcolorbox}
}

\usepackage{thm-restate}
\usepackage{tikz}
\usepackage{pgfplots}
\pgfplotsset{compat=1.17}
\usetikzlibrary{patterns}
\usepgfplotslibrary{fillbetween}
\usetikzlibrary{intersections}
\usetikzlibrary{positioning,fit,calc,arrows.meta}
\usepackage{fix-cm}
\usepackage[ruled,vlined,linesnumbered]{algorithm2e}
\allowdisplaybreaks
\usetikzlibrary{arrows.meta}

\allowdisplaybreaks
\theoremstyle{plain}
\newtheorem{theorem}{Theorem}[section]
\newtheorem{lemma}[theorem]{Lemma}

\newtheorem{proposition}[theorem]{Proposition}
\newtheorem{corollary}[theorem]{Corollary}

\theoremstyle{plain}
\newtheorem{definition}{Definition}[section] 

\newtheorem{remark}[definition]{Remark}

\theoremstyle{plain}
\newtheorem{assumption}{Assumption}[section]

\usepackage[colorlinks=true, linkcolor=blue!70!black, citecolor=blue!70!black, urlcolor=blue, breaklinks=true]{hyperref}
\usepackage{cleveref}
\Crefname{algocf}{Algorithm}{Algorithms}
\crefname{claim}{claim}{claims}

\newcounter{relctr} 
\everydisplay\expandafter{\the\everydisplay\setcounter{relctr}{0}} 

\AtBeginDocument{} 

\newcommand{\squishlist}{
\begin{list}{{{\small{$\bullet$}}}}
{\setlength{\itemsep}{3pt}      \setlength{\parsep}{1pt}
\setlength{\topsep}{1pt}       \setlength{\partopsep}{0pt}
\setlength{\leftmargin}{1em} \setlength{\labelwidth}{1em}
\setlength{\labelsep}{0.5em} } }
\newcommand{\squishend}{  \end{list}}

\newcommand{\fea}{x}
\newcommand{\fearea}{X}
\newcommand{\feanum}{n}
\newcommand{\feaspace}{\mathcal{X}}
\newcommand{\labelspace}{\mathcal{Y}}
\newcommand{\datadist}{\mathcal{D}}
\newcommand{\feaprob}{\lambda}

\newcommand{\bayesprob}{q}

\newcommand{\labelrea}{Y}

\newcommand{\feasubset}{S}
\newcommand{\feasubsetprob}{\lambda}
\newcommand{\feasubsetlabel}{Y}
\newcommand{\feasubsetmean}{\mu}

\newcommand{\givenclassifier}{f}
\newcommand{\classifiernum}{k}

\newcommand{\cellnum}{m}
\newcommand{\partitiondom}{\succ}
\newcommand{\partitiondomeq}{\succeq}
\newcommand{\blackwelldom}{\succ_{\textsc{B}}}

\newcommand{\algo}{\mathcal{A}}
\newcommand{\detclassifierfam}{\classifierfam^{\textsc{det}}}
\newcommand{\partitiondompare}{\classifierfam_{\partitiondom}}
\newcommand{\partitiondomparedet}{\classifierfam_{\partitiondom}^{\textsc{det}}}

\newcommand{\dom}{\succ}

\newcommand{\predicrefine}{\textsc{PredictorRefine}}

\newcommand{\signalspace}{\Sigma}
\newcommand{\kernal}{T}
\newcommand{\signal}{s}
\newcommand{\signalrea}{S}
\newcommand{\predictionrea}{P}
\newcommand{\predictionvalue}{a}
\newcommand{\coarsenedclassifier}{\bar{\givenclassifier}}
\newcommand{\mapkernal}{H}

\newcommand{\signalvec}{\vv}
\newcommand{\signalvecmtx}{\mathbf{V}}

\newcommand{\signalbag}{\mathcal{V}}
\newcommand{\signalbagraw}{\signalbag^{\mathrm{raw}}}
\newcommand{\signalbagred}{\signalbag^{\mathrm{red}}}
\newcommand{\rayeq}{\sim_{\mathrm{ray}}}

\newcommand{\kernalmtx}{\mathbf{H}}

\newcommand{\partition}{\Pi}
\newcommand{\cell}{A}
\newcommand{\cellnew}{C}
\newcommand{\prediction}{p}
\newcommand{\totalcells}{N}

\newcommand{\mymatrix}{\mathbf{M}}
\newcommand{\weightedbayes}{\mathbf{y}}
\newcommand{\weightedpredic}{\mathbf{b}}

\newcommand{\rowspace}{\textsc{Row}}
\newcommand{\myspan}{\textsc{span}}
\newcommand{\classifier}{g}
\newcommand{\classifiernew}{h}
\newcommand{\CDF}{F}

\newcommand{\subsetindi}{\mathbf{1}}
\newcommand{\classifierfam}{\mathcal{G}}
\newcommand{\classifierfamdet}{\classifierfam^{\textsc{det}}}

\newcommand{\scalescalar}{\beta}

\newcommand{\omt}[1]{}

\newcommand{\cc}[1]{\ensuremath{\mathsf{#1}}} 

\newcommand{\prob}[2][]{\mathbb{P}\ifthenelse{\not\equal{}{#1}}{_{#1}}{}\!\left[{\def\givenn{\middle|}#2}\right]}
\newcommand{\expect}[2][]{\mathbb{E}\ifthenelse{\not\equal{}{#1}}{_{#1}}{}\!\left[{\def\givenn{\middle|}#2}\right]}
\newcommand{\indicator}[2][]{\mathbf 1\ifthenelse{\not\equal{}{#1}}{_{#1}}{}\!\left[{\def\givenn{\middle|}#2}\right]}

\newcommand{\R}{\mathbb{R}}

\newcommand{\vv}{\mathbf{v}}
\newcommand{\vu}{\mathbf{u}}

\newcommand{\vz}{\mathbf{z}}
\newcommand{\vzero}{\mathbf{0}}
\newcommand{\supp}{\cc{supp}}

\newcommand{\one}{\boldsymbol{1}}
\newcommand{\primed}{^{\dagger}}
\newcommand{\doubleprimed}{^{\ddagger}}

\newcommand{\reals}{\R}

\newcommand{\xhdr}[1]{\vspace{6pt}\noindent{\bf {#1.}}}
\usepackage{natbib}
\bibpunct{(}{)}{;}{a}{,}{,}

\begin{document}

\title{Refinement as a Service: Algorithmic Predictor Refinement}
\author{
Wei Tang\thanks{Chinese University of Hong Kong. Email: {\tt weitang@cuhk.edu.hk}}
\and
Hanrui Zhang\thanks{Chinese University of Hong Kong. Email: {\tt hanrui@cse.cuhk.edu.hk}}
}
\date{}

\maketitle

\begin{abstract}
    
Prediction aggregation aims to combine information from multiple predictors into a more informative one. 
We study this question in the setting of calibrated predictors, where each prediction must equal the conditional expectation of the quantity being predicted given the predictor's signal. 
Given several calibrated input predictors and the feature distribution, but not the underlying Bayes probabilities, we ask when one can construct refined calibrated predictors that preserve the information in the original predictors and cannot be further refined using the available information.

We formulate calibrated predictors as signaling schemes and define refinement through feature-independent garblings: a predictor refines another if its signal can simulate the other's signal. 
Constructibility is characterized through observable linear information: each signal corresponds to a vector over the feature space, and a new signal is constructible exactly when its vector lies in the linear span of the input signal vectors. 
Under this formulation, we establish a sharp algorithmic picture. 
For deterministic output predictors, bilateral refinement admits a polynomial-time algorithm based on a bipartite graph between the two input signal partitions, while refinement with an arbitrary number of input predictors is $\mathsf{NP}$-hard. 
In contrast, when randomized output predictors are allowed, we give a polynomial-time algorithm for any number of input predictors by decomposing constructible signal vectors into extreme rays of the associated polyhedral cone.\footnote{A conference version of this paper is accepted by The 40th Annual Conference on Neural Information Processing Systems (NeurIPS 2026).}

\end{abstract}




\newpage

\section{Introduction}

{\em Prediction aggregation} \wtedit{(or forecast aggregation)} is the process of combining multiple predictions into a single prediction in a principled way (see, e.g., \citealp{S-61,BG-69,GBR-07,RG-10}).  
The idea is that the aggregate prediction contains information inherent in all input predictions, and is therefore more ``informative'' than any single input prediction.  This very idea, instantiated in different ways, lies behind fundamental methodologies and tools in different branches of computer science / machine learning: In supervised learning, boosting techniques are used to turn weak learners into strong ones \citep{S-90,F-95,FHT-00}; in online learning, multiplicative weights algorithms achieve low regret against the best fixed expert \citep{LW-94,CL-06}; in crowdsourcing and computational social choice, voting and aggregation methods combine multiple individuals' opinions into a single, supposedly representative opinion, based on which social decisions can be made \citep{Retal-10,BCELP-16}; even in LLM architectures, mixture of experts (MoE) models show drastically improved efficiency and scalability while achieving similar, or even better performance compared to dense models \citep{JJNH-91,Detal-24,Lieatal-24}.  All this highlights the power of the clean idea underlying prediction aggregation, which motivates our study in this paper.

At the same time, the above also hints at the {\em immensity} of this very idea: It is deeply rooted in so many different areas, where it manifests in so many different technical forms, to the extent that any serious technical discussion of the idea must refrain from talking about the idea itself --- instead, one must begin by putting the idea into a specific, technical context.  The context of our investigation is {\em refinement of calibrated predictors}.  Below we explain the setup and the keywords (``predictors'', ``calibrated'', and ``refinement'') using a realistic example.

\xhdr{Example: click-through rate (CTR) prediction in online advertising}
A user visits a platform (e.g., Google or Meta), which triggers the display of an ad.  The platform must decide which ad to display from a set of candidates, and the decision depends on the predicted click-through rates (CTRs) of these ads, which are the predicted probabilities that the user would click on each of these ads.  The platform sets up a {\em predictor} for each candidate ad, which maps the user's features (region, language, OS, historical activities on the same platform, etc.) to a numerical prediction.  Typically, this predictor is a machine learning model trained on historical data.  Needless to say, the use of such predictors is omnipresent in all kinds of applications.

\aidelete{In online advertising, we want the predictors to be {\em calibrated}, meaning that conditioning on a particular prediction made by a predictor, the probability that the user would click on the ad is precisely the predicted value \citep{Dawid-82,FV-97}.  In other words, the prediction is always equal to the conditional expectation of the true CTR given the prediction.}
\aiedit{In online advertising, we want the predictors to be {\em calibrated}, meaning that, conditional on each signal generated by a predictor, the probability that the user would click on the ad equals the numerical prediction associated with that signal \citep{Dawid-82,FV-97}. 
We refer to this notion as {\em signal-level calibration}. It implies conventional calibration with respect to numerical predictions, but the converse need not hold when distinct signals are assigned the same numerical prediction.}\footnote{
    Note that this does not mean the predictor has to be perfectly accurate for every feature vector.  \aireplace{In particular, a calibrated predictor could always make a constant prediction equal to the expected CTR over the entire population.}{In particular, a calibrated predictor with a single signal could always make a constant prediction equal to the expected CTR over the entire population.}
}
This ensures that the ads are properly ranked, and the payments that the advertisers are charged are correctly computed \citep{MPZ-13,hetal-14,GoogleExpectedCTR}.  Calibration is also desired in many other scenarios, e.g., medical diagnosis and weather forecasting.

The need for {\em refinement} of predictors arises when two or more platforms hold different information about users.  The idea is the same as in prediction aggregation: to combine information held by different participants and improve the quality of the predictors used by all participants.  To make this concrete, suppose there are $3$ types of users that each represent $1/3$ of the entire population:  Type A, Type B, and Type C.  Two platforms, numbered 1 and 2, each hold a calibrated predictor ($\givenclassifier_1$ and $\givenclassifier_2$, respectively) for the CTR of a particular ad, and they wish to refine their respective predictors by incorporating information from the other platform's predictor.\footnote{
    Note that generally, it makes little sense to combine predictors for different ads, because they are not predicting the same quantity.
}
Further suppose $\givenclassifier_1$ and $\givenclassifier_2$ make calibrated predictions in the following ways:
\begin{itemize}
    \item $\givenclassifier_1$ predicts $0.3$ for Type A and Type B users, and $0.7$ for Type C users.
    \item $\givenclassifier_2$ predicts $0.1$ for Type A users, and $0.6$ for Type B and Type C users.
\end{itemize}
Then, because both predictors are calibrated, one may deduce the following about the true CTR: All Type A and Type B users on average have a CTR of $0.3$, and all Type A users on average have a CTR of $0.1$.  As a result, all Type B users on average have a CTR of $0.3 \times 2 - 0.1 = 0.5$.  Now, after exchanging information, both platforms know {\em everything} about the CTR of each type of users: Type A: $0.1$, Type B: $0.5$, Type C: $0.7$.  This immediately gives a new calibrated predictor, which refines both $\givenclassifier_1$ and $\givenclassifier_2$, because it makes ``finer'' predictions than both input predictors.  Here, the notion of finer predictions is crucial: It ensures that the new predictor {\em never performs worse} than the corresponding original one for any user type, and {\em sometimes performs strictly better}, for some user types.  As such, the new predictor achieves, in economic terms, a {\em Pareto improvement} over the original one.

In general, refinement in the above sense is nontrivial.  In fact, refinement may not be possible at all even if the input predictors contain very different information.  Consider, for example, $4$ types of users: A, B, C, and D.  $\givenclassifier_1$ predicts $0.3$ for Type A and Type B, and $0.7$ for Type C and Type D; $\givenclassifier_2$ predicts $0.4$ for Type A and Type C, and $0.6$ for Type B and Type D.  One may check that there is no way to refine either input predictor based on information contained in the other.  This leads us to the main questions addressed in the paper:

\begin{quote}
    Given a collection of input calibrated predictors, how do we tell whether it is possible to refine (one of) these predictors using information from other input predictors?  Taking a step further, how do we refine all input predictors to the maximum extent possible, such that no predictor can be further refined given all information available to us?
\end{quote}

\subsection{Setting up the Stage}

\xhdr{Abstraction of calibrated predictors as signaling schemes}
To faciliate our investigation and simplify the presentation, throughout the paper, we treat calibrated predictors in an abstract way, essentially as signaling schemes.\footnote{
    We will henceforth assume all predictors under discussion are calibrated, and omit the modifier ``calibrated'' whenever possible.
}
The idea is the following: A predictor essentially ``partitions'' the feature space and ``packs'' feature vectors into indistinguishable (as perceived by the predictor) subsets.  Within each indistinguishable subset, the predictor, being calibrated, must predict the conditional expectation of the quantity to be predicted (e.g., a predictor that predicts $0.3$ for features A and B, and $0.7$ for feature C induces a partition where A and B form an indistinguishable subset, and C forms another; the predictions, $0.3$ and $0.7$, are the respective conditional expectations within each subset).  In other words, the prediction made by a predictor on an indistinguishable subset is uniquely determined by what is contained in this subset.  As such, a predictor can be fully summarized by a signaling scheme, which, on each indistinguishable subset, produces a signal rather than a numerical prediction.  

\wtedit{
This viewpoint is consistent with the standard informational interpretation in forecast aggregation literature, where reported forecasts are typically treated as posterior beliefs under the forecasters' information. Our focus is different: we ask {\em when} and {\em how} the information contained in several calibrated predictors can be combined into a constructible refinement (see, e.g., \citealp{S-61,DF-83,C-89,ABS-18,GHHKSY-25} and \Cref{subsec:related} for more discussions).
}
More formally, we model each predictor $\givenclassifier$ as a (possibly randomized) mapping from the feature space to a discrete signal space without any structure (e.g., the predictor mentioned a few sentences ago now becomes something like $\{\{\text{A}, \text{B}\} \mapsto \text{Signal 1}, \{\text{C}\} \mapsto \text{Signal 2}\}$, with the numerical values of the predictions abstracted away in the technical framework. 
As we will see, this formulation better reveals the combinatorial nature of the problem being investigated.
\aiedit{This abstraction does not mean that numerical predictions are unavailable. We treat signals as the primitive objects for comparing the information carried by different predictors, while {\em conceptually} keeping their associated numerical predictions as observed quantities. 
Under signal-level calibration, these values are determined by the signal structure and the underlying data distribution. Although the latter is unknown to the aggregator, the observed input predictions allow us to compute calibrated numerical predictions for constructible refined signals.}

\xhdr{Information-theoretical formulation of ``refinement''}
Recall that our goal is to {\em refine} input predictors, which in particular means that, roughly speaking, the information contained in the refined predictor should {\em ``subsume''} that contained in the corresponding input predictor {\em ``as a subset''}.  This intuition is captured by the following simple and natural notion of refinement: Given two predictors $\givenclassifier$ and $\classifier$, if one could simulate the predictions (which are abstract signals) of $\givenclassifier$ based solely on the predictions of $\classifier$ in a blackbox way, then $\classifier$ refines $\givenclassifier$ --- or in other words, $\classifier$ (information-theoretically) dominates $\givenclassifier$.  For example, consider $f: \{\{\text{A}, \text{B}\} \mapsto \text{Signal 1}, \{\text{C}\} \mapsto \text{Signal 2}\}$ and $g: \{\{\text{A}\} \mapsto \text{Signal 3}, \{\text{B}\} \mapsto \text{Signal 4}, \{\text{C}\} \mapsto \text{Signal 5}\}$.  Given the predictions made by $\classifier$, one could map Signal 3 and Signal 4 into Signal 1, and Signal 5 into Signal 2, thereby recovering the predictions made by $\givenclassifier$ exactly.  As a result, $\classifier$ refines / dominates $\givenclassifier$.  Now we can restate our goal a bit more formally: Given $\classifiernum$ input predictors $\givenclassifier_1, \dots, \givenclassifier_\classifiernum$, we want to find ``constructible'' (elaborated below) predictors $\classifier_1, \dots, \classifier_\classifiernum$, such that (1) $\classifier_i$ dominates $\givenclassifier_i$, and (2) among all constructible predictors, no predictor $h \ne \classifier_i$ dominates $\classifier_i$, for each $i \in [k]$.  That is, we want to maximally refine each $\givenclassifier_i$ into $\classifier_i$ in the sense discussed above.

\xhdr{Constructible predictors and the bag-of-vectors representation}
One thing remains unclear in the above formulation: Exactly which predictors are constructible?  It turns out there is a natural definition for that too: A predictor is constructible, iff the conditional expectation corresponding to each possible prediction can be {\em uniquely determined} given the input predictors (see further discussion and examples below for concreteness).  This is a semantic definition, and we would like a more syntactic and operable version of it, which involves representing predictors as bags of vectors.  We will demonstrate the idea by means of examples.

Consider the two predictors in our first introductory example (in their abstract forms): $\givenclassifier_1: \{\{\text{A}, \text{B}\} \mapsto \text{Signal 1}, \{\text{C}\} \mapsto \text{Signal 2}\}$, $\givenclassifier_2: \{\{\text{A}\} \mapsto \text{Signal 3}, \{\text{B}, \text{C}\} \mapsto \text{Signal 4}\}$.  We already know that the conditional mean on all $3$ feature vectors can be uniquely determined from $\givenclassifier_1$ and $\givenclassifier_2$.  To formalize this fact, we write both input predictors as bags of vectors.  For $\givenclassifier_1$, we write $(1, 1, 0)$ for Signal 1, because the corresponding subset contains A (``Feature 1'') and B (``Feature 2''), but not C (``Feature 3''); similarly, we write $(0, 0, 1)$ for Signal 2.  Accordingly, $\givenclassifier_1$ is represented by the associated bag of vectors $\signalbag(\givenclassifier_1) = \{(1, 1, 0), (0, 0, 1)\}$, where each vector corresponds to one prediction, and the sum of all vectors is the all-one vector.  Similarly, $\givenclassifier_2$ is represented by $\signalbag(\givenclassifier_2) = \{(1, 0, 0), (0, 1, 1)\}$.  Note that a predictor, modulo permutations of the signals (which is immaterial), corresponds bijectively to its bag-of-vectors representation.

As before, assume that the $3$ features are equally frequent in the population (which is in fact without loss of generality).  Now, to see why the conditional expectations on all features (and in particular, on feature B) are uniquely determined, simply observe that they can be pinned down by solving the following (overdetermined but guaranteed to be consistent) linear system:
\begin{align*}
    (1, 1, 0) \cdot {\bf p} & = \|(1, 1, 0)\|_1 \cdot p_1 \\
    (0, 0, 1) \cdot {\bf p} & = \|(0, 0, 1)\|_1 \cdot p_2 \\ 
    (1, 0, 0) \cdot {\bf p} & = \|(1, 0, 0)\|_1 \cdot p_3 \\
    (0, 1, 1) \cdot {\bf p} & = \|(0, 1, 1)\|_1 \cdot p_4,
\end{align*}
where ${\bf p} = (p_A, p_B, p_C)^\top$ is the conditional expectations on individual features, and $p_1, \dots, p_4$ are the conditional expectations corresponding to the predictions (Signal 1, \dots, Signal 4).  The first two equations correspond to $\signalbag(\givenclassifier_1)$, and the last two correspond to $\signalbag(\givenclassifier_2)$.  $p_B$ in particular can be uniquely determined, because the vector $(0, 1, 0)$ is in the linear span of $\signalbag(\givenclassifier_1) \cup \signalbag(\givenclassifier_2)$.  Accordingly, the predictor $g: \{\{\text{A}\} \mapsto \text{Signal 5}, \{\{\text{B}\} \mapsto \text{Signal 6}, \{\text{C}\} \mapsto \text{Signal 7}\}$ is constructible, because all $3$ vectors corresponding to the $3$ predictions, $(1, 0, 0)$, $(0, 1, 0)$, and $(0, 0, 1)$, are in the span of $\signalbag(\givenclassifier_1) \cup \signalbag(\givenclassifier_2)$.  \aidelete{Note that in the above reasoning, the specific values of $p_1, \dots, p_4$ are entirely irrelevant, which is also why we abstract them away in our definition of predictors.}\aiedit{Whether these signal vectors lie in the span depends only on the input signal structures. The observed numerical values $p_1,\dots,p_4$ are then used to solve the linear system and compute the calibrated predictions attached to the refined signals.}

The bag-of-vectors representation can be naturally generalized to capture randomized predictors, where each vector is not necessarily $0$-$1$.  It is not hard to show that a prediction whose corresponding vector is ${\bf v}$ leads to a uniquely determined conditional expectation, iff $\bf v$ is in the span of the union of the bag-of-vectors representations of all input predictors.  In later sections of the paper, we directly define constructible predictors based on the above observation, in order to hide less relevant details and stay focused on the combinatorial nature of the problem.  We also remark that the bag-of-vectors representation also gives a more operable equivalent definition of refinement: $\classifier$ refines / dominates $\givenclassifier$ iff $\signalbag(\classifier) \supset \signalbag(\givenclassifier)$.

\subsection{Results and Techniques}

The above observations already give operable ways of {\em checking} whether a given group of predictors $\classifier_1, \dots, \classifier_\classifiernum$ are constructible given $\givenclassifier_1, \dots, \givenclassifier_\classifiernum$, as well as whether each $\classifier_i$ refines $\givenclassifier_i$.  However, our goal is to {\em find} constructible $\classifier_1, \dots, \classifier_\classifiernum$ that {\em maximally} refine $\givenclassifier_1, \dots, \givenclassifier_\classifiernum$, respectively.  As we will see, the latter task involves significantly more technical effort.

\xhdr{Refinement into deterministic predictors}
We first consider a natural restriction of the problem, where we further require the output predictors to be deterministic.  This makes sense in scenarios where randomized predictors are infeasible, e.g., when we need a stable prediction for each feature.  Our results establish a dichotomy between {\em bilateral} refinement and {\em multilateral} refinement:

\begin{theorem}[informal]
    Restricted to deterministic predictors:
    \begin{itemize}
        \item When there are only $2$ input predictors (corresponding to bilateral agreements to share information and refine each participant's predictor), there is a polynomial-time algorithm for predictor refinement.
        \item With an arbitrary number $\classifiernum$ of input predictors (corresponding to multilateral agreements to pool information held by all participants together, which is used to refine each participant's predictor), predictor refinement is $\mathsf{NP}$-hard.
    \end{itemize}
\end{theorem}

Technically, our algorithm for bilateral refinement reveals the curious, graph-theoretical nature of the problem.  The algorithm works by first building an undirected bipartite graph, where the two sides represent the two input predictors, $\givenclassifier_1$ and $\givenclassifier_2$.  Each $0$-$1$ vector in the bag-of-vectors representations of $\givenclassifier_1$ and $\givenclassifier_2$ corresponds to a vertex in the graph.  Two vertices are adjacent, iff their inner product is positive, i.e., iff the corresponding sets of features intersect.  For example, consider again $\givenclassifier_1$ and $\givenclassifier_2$ where $\signalbag(\givenclassifier_1) = \{(1, 1, 0), (0, 0, 1)\}$ and $\signalbag(\givenclassifier_2) = \{(1, 0, 0), (0, 1, 1)\}$.  The vertices corresponding to $(1, 1, 0) \in \signalbag(\givenclassifier_1)$ and $(1, 0, 0) \in \signalbag(\givenclassifier_2)$ are adjacent, as are those corresponding to $(1, 1, 0) \in \signalbag(\givenclassifier_1)$ and $(0, 1, 1) \in \signalbag(\givenclassifier_2)$.  The graph is bipartite because no two vectors in $\signalbag(\givenclassifier_1)$ or $\signalbag(\givenclassifier_2)$ can possibly overlap on a nonempty subset of the feature space (which would lead to an ill-defined predictor).

Given the bipartite graph, which can be built in polynomial time, the algorithm then identifies all bottleneck vertices in all connected components of the graph --- a vertex is a bottleneck vertex, iff the connected component containing the vertex breaks apart when this vertex is removed.  These vertices are crucial because, as we prove, $\givenclassifier_1$ or $\givenclassifier_2$ can be refined iff there exists a bottleneck vertex in the corresponding graph.  Given all the bottleneck vertices, the algorithm then ``blows up'' each of them, meaning that each bottleneck vertex is split into two new vertices, whose sum (in terms of the corresponding $0$-$1$ vectors) is the original bottleneck vertex.  The two new vertices inherit disjoint subsets of the edges connected to the original vertex, and thus we can ensure that the connected component containing the original bottleneck vertex breaks apart into two new connected components, containing the two new vertices respectively.  As an overall result, the algorithm essentially breaks the original graph into its bi-connected components, which is still a bipartite graph.  This new graph corresponds to two new predictors $\classifier_1$ and $\classifier_2$, which are provably constructible, and each maximally refine the corresponding input predictor.

The hardness result for multilateral refinement, on the other hand, is based on the following intuition: Checking whether a particular target vector can be refined (i.e., split into two vectors in the linear span of the input bags of vectors) is closely connected to exact covering, in the sense that we need to find a linear combination of input vectors, with coefficients $1$ and $-1$, such that the irrelevant parts of these vectors cancel out exactly, and the remaining part is ``contained'' (in terms of locations of $1$'s) in the target vector.  Indeed, we establish the hardness through a reduction from a variant of the exact cover problem, which is known to be $\mathsf{NP}$-hard.

\xhdr{Refinement into randomized predictors}
We then investigate the original, unrestricted problem, where output predictors are allowed to be randomized.  In this regime, we present a clean positive result:

\begin{theorem}[informal]
    When randomization is allowed, there is a polynomial-time algorithm for predictor refinement with any number of input predictors.
\end{theorem}

The algorithm has a slight algebraic flavor (since it operates with general real vectors beyond $0$-$1$), but its core is still combinatorial.  First, let us consider the following easier task, which helps build intuition, and also serves as a subroutine in our complete algorithm: Given a target vector from a collection of $m$ input vectors, check whether the target vector can be split into two nonnegative vectors that are both in the span of the input vectors, where neither of the two vectors can be a scaled version of the target vector.  In fact, as we argue, this is equivalent to checking whether there is a nonnegative and nonzero vector in the span of the input vectors, whose support is a {\em strict} subset of that of the target vector --- below we call such a vector a {\em refiner}.  This is possible in polynomial time through repeatedly solving a polynomial number of LPs, where we enumerate one coordinate to be dropped from the support of the target vector and check feasibility by solving an LP.

In fact, one can strengthen the above procedure such that it finds a refiner whose support is minimal, i.e., there exists no refiner whose support is a strict subset of the refiner that we find.  To see why this is possible, observe that one could simply run the procedure above iteratively, where in each iteration, we use the refiner found in the previous iteration as the target vector in the current iteration.  This stops when we find a refiner that does not have its own refiner, which means it must have a minimal support.  Note that following the chain of reasoning above, such a minimal refiner cannot be split into two essentially different nonnegative vectors both in the span of the input vectors.

With the more powerful subroutine, we are ready to construct the complete algorithm for predictor refinement.  A straightforward idea is to apply the subroutine above to each vector in the bag-of-vectors representation of each input predictor.  If it successfully finds a refiner, then we can use it to split the target vector into two new vectors, thus refining an input predictor; if no refiner is found, then it must be the case that the target vector cannot be further refined.  We stop when no vector in any bag of vectors admits a refiner, at which point we know the current predictors corresponding to these bags of vectors already maximally refine the input predictors.  The key issue here is that we need a way to bound the total number of splits: When we split a vector, it could be the case that one or both of the new vectors created by the split still admit a refiner, and can therefore be further split.  This is precisely why we want to make sure that whenever we split a vector, we do so using a minimal refiner as found by the subroutine described above.  As we argue, this ensures that the two new vectors --- referred to as the refiner and the residue --- are both ``simpler'' than the target vector being split.  In particular, the refiner cannot be further split, and the residue has a strictly smaller support than the target vector.  So, to bound the total number of splits, we consider the following potential function: the sum of the size of the support of every vector that {\em might} still be split.  The potential is always integral, nonnegative, and polynomial in the relevant parameters.  After each split, it decreases by at least one, which immediately implies that the total number of splits is polynomial.

\xhdr{Application: refinement as a service}
\wtedit{Our algorithmic results can be used directly when participants share their predictors and run
one of our refinement algorithms. They also suggest a service-oriented interpretation, motivated by the growing use of machine-learning services and prediction APIs \citep{TramerEtAl2016,Bommasaniwtal-21,GoogleVertexAICustomPrediction}.
A platform may hold a highly optimized general-purpose predictor, while an agent may hold a specialized predictor for a particular domain. The platform can run our algorithm to incorporate all constructible information from its predictor into the agent's specialized predictor, producing a refined predictor that is more informative while preserving the agent's domain-specific expertise. Conversely, upon agreement, the platform can refine its
own predictor using the agent's information, or add the agent's predictor to its available bag of signal vectors. In this sense, refinement can be viewed as a service for improving calibrated predictors without requiring access to the underlying Bayes probabilities.}

\subsection{Further Related Work}
\label{subsec:related}

Prediction aggregation, or forecast aggregation, and closely related formulations have been studied across several communities, including machine learning, statistics, economics, and theoretical computer science.
Early foundational work includes opinion pooling and forecast combination \citep{S-61,BG-69}, as well as classical reviews of forecast combination \citep{C-89}.
The broader literature has since developed along different methodological directions.

In machine learning, aggregation is often studied through ensemble methods, where multiple predictors/classifiers are combined to improve predictive performance. 
Classical examples include stacking, bagging, boosting, and random forests \citep{W-92,B-96,S-90,F-95,FS-97,FHT-00,B-01,D-00}. 
A related online-learning perspective studies prediction with expert advice, where multiplicative-weights-type algorithms combine experts' predictions while achieving low regret relative to the best fixed expert
\citep{LW-94,CL-06}.

In statistics, forecast aggregation has been studied both axiomatically and probabilistically.
The axiomatic literature imposes desiderata on aggregation rules, such as unanimity preservation and variants of independence, and characterizes rules such as linear or externally Bayesian pooling under corresponding assumptions \citep{AW-80,G-84,DL-16}.
The probabilistic and Bayesian literature instead models how
forecasts are generated from underlying signals and then derives aggregation rules by Bayesian updating or parametric estimation; see, for example, \citet{SBFMTU-14,FCK-15,EPSU-16,SPU-16,RG-10}.

A common feature of much of this literature is that the reported forecasts are interpreted as probabilistic beliefs induced by the forecasters' information \citep{B-82, GBR-07, SPU-16,ABS-18,GHHKSY-25}. Under this interpretation, a forecast is calibrated with respect to the forecaster's signal: it represents the posterior mean of the target conditional on the information available to that forecaster. Our focus on calibrated predictors is therefore aligned with this standard informational view of forecasts.

In economics and theoretical computer science, a recent line of work takes a robust
perspective and asks for aggregators that perform well across large classes of possible
information structures, especially when the correlation structure among experts' signals is
unknown or misspecified \citep{ABS-18,DIL-21,LR-22,NR-22,KWW-24,GHHKSY-25,GK-25,FMNW-25,CPT-26}.
A complementary strand leverages second-order or higher-order information, such as agents' beliefs about others' answers, to improve aggregation in finite populations \citep{P-04,PSM-17,PS-19,WLC-21,WMH-22,PCK-24,APSTX-25}.

Our work differs from these lines in both its objective and its information constraint. 
Rather than designing a rule that maps several reported forecasts into a single aggregate forecast, or evaluating such a rule under a loss or regret criterion, we study when and how calibrated predictors can be refined using only the observable information contained in the input predictors, without access to the underlying Bayes probabilities. This leads to a distinct constructibility problem: the goal is to move from numerical aggregation to information refinement, and to characterize and compute maximally informative constructible refinements under an information-theoretic dominance relation.
\wtedit{Another closely related work by \citet{TZ-26} studies aggregation of calibrated numerical forecasts using an analogous observable linear space and nonnegative cone. 
However, their refinement relation compares posterior prediction distributions through the Blackwell order \citep{B-53} for the binary outcome \(Y\). 
Our \Cref{defn:partition dom} instead requires a feature-independent garbling to reproduce the signal distribution conditional on every feature \(x\in\mathcal X\).}  Accordingly, our notion of refinement is typically stronger than theirs, in the sense that if $f$ refines $g$ under our definition, then $f$ must also dominate $g$ under the Blackwell order adopted by \citet{TZ-26}; the converse is generally not true.  The different notions also lead to fundamentally different technical problems.

\section{Preliminary}
\label{sec:prelim}

\newcommand{\predicrefinelink}{\hyperref[prob:predictor-refine]{\predicrefine}}

We begin by introducing the basic notation for the {\em predictor refinement} problem in the binary classification setting.
Let $\feaspace$ be a finite feature space,
and let $\labelspace = \{0, 1\}$ be the binary label space.
The data are generated according to a joint distribution $\datadist \in \Delta(\feaspace\times\labelspace)$.

A possibly randomized predictor $\givenclassifier = (\signalspace_\givenclassifier, \kernal_\givenclassifier, \prediction_\givenclassifier)$ consists of a signal space $\signalspace_\givenclassifier$, a stochastic kernel $\kernal_\givenclassifier: \feaspace \to \Delta(\signalspace_\givenclassifier)$, and a prediction map $\prediction_\givenclassifier:\signalspace_\givenclassifier\to[0,1]$. 
The kernel $\kernal_\givenclassifier(\cdot\mid \fea)$ specifies, for each feature $\fea$, a distribution over signals, with $\kernal_\givenclassifier(\signal\mid \fea)$ denoting the conditional probability of realizing signal $\signal$ given the feature realization $\fearea=\fea$; the prediction map then assigns each realized signal $\signal\in\signalspace_\givenclassifier$ a numerical prediction $\prediction_\givenclassifier(\signal)\in[0,1]$.
Given a feature realization $\fearea=\fea$, the predictor first draws a signal
$\signal \sim \kernal_\givenclassifier(\cdot \mid \fea)$ and then assigns a prediction $\prediction_\givenclassifier(\signal)$ to this signal.
\aidelete{Together with the data distribution $\datadist$, the predictor $\givenclassifier$ induces a joint distribution
over predictions and labels, denoted by
$\Gamma_\givenclassifier \in \Delta([0,1]\times \labelspace)$.
Let $\CDF_\givenclassifier$ denote the CDF of the induced marginal distribution of
the predictions $\prediction$.
Throughout this work, we focus on perfectly calibrated predictors. 
Formally, a predictor $\givenclassifier$ is perfectly calibrated if the following condition holds:
\begin{align*}
    \expect[(\prediction,\labelrea)\sim \Gamma_\givenclassifier]{\labelrea \mid \prediction}
    = \prediction,
    \quad
    \text{for every } \prediction \in \supp(\CDF_\givenclassifier)~.
\end{align*}
As we can see, under perfect calibration w.r.t.\ the data distribution $\datadist$, the prediction assigned to each signal is uniquely determined once the signal space $\signalspace_\givenclassifier$ and kernel $\kernal_\givenclassifier$ are fixed. 
Consequently, we often omit the prediction map and write $\givenclassifier=(\signalspace_\givenclassifier,\kernal_\givenclassifier)$.
A predictor $\givenclassifier$ is deterministic if $\kernal_\givenclassifier(\cdot\mid \fea)$ is a point mass for every feature $\fea\in\feaspace$.}
\aiedit{Let $\signalrea_\givenclassifier$ denote the random signal generated by predictor $\givenclassifier$, and let $\predictionrea_\givenclassifier=\prediction_\givenclassifier(\signalrea_\givenclassifier)$ denote its numerical prediction. Throughout this paper, we focus on perfectly calibrated predictors at the signal level. Formally, a predictor $\givenclassifier$ is {\em signal-level calibrated} with respect to the data distribution $\datadist$ if
\begin{align}
    \expect{\labelrea\mid \signalrea_\givenclassifier=\signal}
    &=\prediction_\givenclassifier(\signal)
    \quad\text{for every }\signal\in\signalspace_\givenclassifier
    \text{ with }\prob{\signalrea_\givenclassifier=\signal}>0~.
    \label{eq:signal-level-calibration}
\end{align}
Under signal-level calibration, the numerical prediction associated with each positive-probability signal is uniquely determined once the data distribution $\datadist$, signal space $\signalspace_\givenclassifier$, and signal kernel $\kernal_\givenclassifier$ are fixed. Consequently, when discussing the information structure of a predictor, we often omit its prediction map and write $\givenclassifier=(\signalspace_\givenclassifier,\kernal_\givenclassifier)$. This is only a notational simplification: the numerical predictions associated with the input signals are available to the aggregator, whereas the underlying Bayes probabilities are not. A predictor $\givenclassifier$ is deterministic if $\kernal_\givenclassifier(\cdot\mid \fea)$ is a point mass for every feature $\fea\in\feaspace$.}

\begin{remark}
\aiedit{We note that the condition in Eqn.~\eqref{eq:signal-level-calibration} implies conventional {\em prediction-level calibration}, under which the expected label conditional on a numerical prediction equals that prediction. 
Indeed, by the law of iterated expectations,
\begin{align*}
    \expect{\labelrea\mid\predictionrea_\givenclassifier}
    =\expect{\expect{\labelrea\mid\signalrea_\givenclassifier}\mid\predictionrea_\givenclassifier}
    =\expect{\prediction_\givenclassifier(\signalrea_\givenclassifier)\mid\predictionrea_\givenclassifier}
    =\predictionrea_\givenclassifier~.
\end{align*}
The converse need not hold when distinct signals are assigned the same numerical prediction. For example, consider two equally likely signals with conditional label means $0$ and $1$, respectively, and suppose that both signals are assigned the numerical prediction $1/2$. The resulting predictor satisfies prediction-level calibration, because it always predicts the population mean $1/2$, but neither signal is individually calibrated.}
\end{remark}

We formulate predictors as signal schemes, rather than simply as maps from features to numerical predictions, because calibrated numerical values alone may collapse signals that are informationally distinct. Since refinement is meant to compare the information carried by predictors, we keep realized signals as the primitive objects and recover their numerical values only through calibration.

For example, suppose two features \(\fea_1\) and \(\fea_2\) both have calibrated value \(1/2\). Consider one predictor that uses two distinct signals, one for \(\fea_1\) and one
for \(\fea_2\), and another predictor that pools \(\{\fea_1,\fea_2\}\) into a single signal. Both
predictors always output the same numerical prediction, \(1/2\). However, they are not
informationally equivalent: the first predictor can simulate the second by forgetting which
signal was realized, while the second cannot recover the distinction between \(\fea_1\) and
\(\fea_2\). Thus, representing predictors only by their numerical predictions would identify
these two predictors, whereas the signal-based representation correctly distinguishes their
refinement content.
\aiedit{This example further illustrates why our refinement analysis treats signals as primitive objects: even when two distinct signals have the same calibrated numerical prediction, distinguishing them may convey additional information about the underlying feature.}

With this signal-based representation in place, we next define the dominance relation used to compare predictors.

\xhdr{Informational dominance} 
We now define the information-theoretic dominance notion, which we refer to as the informational dominance, used throughout the paper. 
\begin{definition}[Informational dominance]
\label{defn:partition dom}
Let $\givenclassifier = (\signalspace_\givenclassifier, \kernal_\givenclassifier), \classifier = (\signalspace_\classifier, \kernal_\classifier)$
be two (possibly) randomized predictors. 
We say that $\givenclassifier$ informationally-dominates $\classifier$, denoted by $\givenclassifier \partitiondom g$,
if and only if the
following two conditions hold:
\begin{itemize}
    \item 
    there exists a stochastic kernel 
    $\mapkernal : \signalspace_\givenclassifier \to \Delta(\signalspace_\classifier)$
    such that, for every feature $\fea \in \feaspace$,
    we have $\kernal_\classifier(\cdot \mid \fea)
    =
    \sum\nolimits_{\signal \in \signalspace_\givenclassifier} \mapkernal(\cdot \mid \signal)\, \kernal_\givenclassifier(s \mid \fea)$.
    \item 
    there does not exist any stochastic kernel $\mapkernal:\signalspace_\classifier\to \Delta(\signalspace_\givenclassifier)$
    such that, for every feature $\fea\in \feaspace$,
    we have $\kernal_\givenclassifier(\cdot\mid \fea)=\sum_{\signal\in \signalspace_\classifier} \mapkernal(\cdot\mid \signal)\, \kernal_\classifier(\signal\mid \fea)$.
\end{itemize}
Equivalently, conditional on $\fearea=\fea$, one may first draw the hidden signal $\signal_\givenclassifier$ from
$\kernal_\givenclassifier(\cdot \mid \fea)$, and then pass $\signal_\givenclassifier$ through the feature-independent noisy channel $\mapkernal$ to obtain a new signal $\signal_\classifier$; the resulting conditional law of $\signal_\classifier$ is exactly $\kernal_\classifier(\cdot \mid \fea)$.
\end{definition}

\wtedit{Intuitively, the dominance relation \(\givenclassifier \partitiondom \classifier\) means that the signal generated by predictor \(\givenclassifier\) contains all the information needed to reproduce the signal generated by predictor \(\classifier\). 
Indeed, after observing the signal \(\signal_\givenclassifier\), one can apply the same feature-independent post-processing rule \(\kernal\), regardless of the underlying feature \(\fea\), to simulate a signal with exactly the same conditional distribution as \(\givenclassifier\)'s signal. Thus, replacing \(\givenclassifier\) by \(\classifier\) does not lose any information relevant to the signal structure of \(\classifier\). 
The second condition makes the relation strict: predictor \(\classifier\) cannot, in turn, simulate the predictor \(\givenclassifier\) through any feature-independent post-processing rule.}


\xhdr{The predictor-refinement problem}
Fix the feature marginal distribution $\feaprob\in\Delta(\feaspace)$ induced by the data distribution $\datadist$,
and fix $\classifiernum$ predictors $(\givenclassifier_j)_{j\in[\classifiernum]}$.
A predictor $\classifier$ is said to be constructible from $\feaprob$ and $(\givenclassifier_j)_{j\in[\classifiernum]}$  
if it can be determined solely from this information; that is, if there exists a mapping $\algo$ such that $\classifier = \algo(\feaprob, (\givenclassifier_j)_{j\in[\classifiernum]})$.
Equivalently, $\classifier$ can be computed using only the information contained in $\feaprob$ and $(\givenclassifier_j)_{j\in[\classifiernum]}$, without access to the unknown Bayes probabilities $\{\prob{\labelrea=1\mid \fea}\}_{x\in\feaspace}$.
We write $\classifierfam$ for the set of all such constructible predictors. When there is no ambiguity, we suppress the dependence of $\classifierfam$ on $\feaprob$ and $(\givenclassifier_j)_{j\in[\classifiernum]}$ in the notation.
We next define the subset of constructible predictors that are undominated under our informational dominance.
\begin{definition}[Undominated predictors]
Fix a feature distribution $\feaprob = (\feaprob_i)$ and fix  $\classifiernum$ predictors $(\givenclassifier_j)_{j\in[\classifiernum]}$.
Let $\partitiondompare \subseteq \classifierfam$ be the set of undominated constructible predictors:
\begin{align*}
    \partitiondompare
    \triangleq
    \{\classifier\in\classifierfam:\ \text{there does not exist }\classifier\primed \in \classifierfam 
    \text{ such that }\classifier\primed \partitiondom \classifier\}~.
\end{align*}
We also denote by $\detclassifierfam\subset \classifierfam$ the set of all constructible deterministic predictors, and let
\begin{align*}
    \partitiondomparedet \triangleq
    \{\classifier\in\detclassifierfam:\ \text{there does not exist }\classifier\primed \in \detclassifierfam 
    \text{ such that }\classifier\primed \partitiondom \classifier\}~.
\end{align*}
\end{definition}
In other words, the set $\partitiondompare$ (resp.\ $\partitiondomparedet$) consists of those constructible  predictors that are not dominated, within the class $\classifierfam$ (resp.\ $\detclassifierfam$).
We are now ready to formalize the predictor refinement problem.

\phantomsection
\label{prob:predictor-refine}
\mybox{
$\predicrefine$:\\
\textbf{Input:} a feature distribution $\feaprob=(\feaprob_i)_{i\in[\feanum]}$ and $\classifiernum$ predictors
$(\givenclassifier_j)_{j\in[\classifiernum]}$. \\
\textbf{Question:} output a collection of predictors $(\classifier_j)_{j\in[\classifiernum]}$ such that for every $j\in[\classifiernum]$, 
\begin{itemize}
    \item 
    if there exists a predictor $\classifier\in \partitiondompare$ such that $\classifier \partitiondom \givenclassifier_j$, then 
    $\classifier_j \in \partitiondompare$ and $\classifier_j \partitiondom \givenclassifier_j$;
    \item 
    otherwise $\classifier_j = \givenclassifier_j$.
\end{itemize}
}

\aiedit{Throughout the paper, each input predictor is assumed to satisfy signal-level calibration. The aggregator observes the input signal kernels together with their associated calibrated numerical predictions, but does not observe the feature-level Bayes probabilities. The signal kernels determine which refinements are informationally constructible, while the observed numerical predictions are used to recover the calibrated values associated with the refined signals.}


Equivalently, the problem 
{\predicrefinelink}
asks, for each initial predictor $\givenclassifier_i$, to return an undominated strict improvement whenever such an improvement exists, and to return the original predictor $\givenclassifier_i$ unchanged otherwise. Importantly, the aggregator observes only the feature distribution $\feaprob$ and the initial predictors $(\givenclassifier_j)_{j\in[\classifiernum]}$, and does not observe the underlying Bayes probabilities.

\xhdr{Additional notations}
We write the feature space as $\feaspace = \{\fea_i\}_{i\in[\feanum]}$, and let $\feaprob = (\feaprob_i)_{i\in[\feanum]}$ be the feature distribution  where $\prob{\fearea = \fea_i}=\feaprob_i$ for $i\in[\feanum]$ denotes the probability mass of feature $\fearea = \fea_i$.
For each feature $\fea_i$, we denote its corresponding Bayes probability for the label $\labelrea = 1$ as $\bayesprob_i \triangleq \prob{\labelrea=1\mid \fearea = \fea_i}\in[0,1]$.
\wtedit{For any nonzero vector $\signalvec\in\R_+^{\feanum}$, we slightly abuse notation and define
\begin{align*}
    \feasubsetprob(\signalvec)
    \triangleq
    \sum\nolimits_{i\in[\feanum]}\feaprob_i \signalvec_i,
    \feasubsetlabel(\signalvec)
    \triangleq
    \sum\nolimits_{i\in[\feanum]}\feaprob_i \bayesprob_i \signalvec_i,
    \feasubsetmean(\signalvec)
    \triangleq
    \frac{\feasubsetlabel(\signalvec)}{\feasubsetprob(\signalvec)}~.
\end{align*}
For some nonempty $\feasubset\subseteq[\feanum]$, we also write $\subsetindi_{\feasubset}$ as the set-indicator vector.}

\section{Observable Linear Information}

In this section, we recast the signal-kernel representation for the predictors from \Cref{sec:prelim} into an equivalent linear form (see \Cref{subsec:bag-of-vectors}), which gives rise to
a formal definition of constructible predictors from observable information $(\feaprob, (\givenclassifier_j)_{j\in[\classifiernum]})$ (see \Cref{subsec:constructible and linear-infor}).



\subsection{Representing Predictors via Bag-of-Vectors}
\label{subsec:bag-of-vectors}
In this subsection, we reformulate the predictors and the informational dominance via the bag-of-vectors. 
Throughout this section, by a \emph{bag} we mean a finite multiset.
\begin{definition}[Raw bag-of-vectors representation]
\label{def:raw-bag}
Let $\givenclassifier=(\signalspace_{\givenclassifier},\kernal_{\givenclassifier})$ be a finitely supported predictor.
Relabel the signal space as $\signalspace_{\givenclassifier} =[\cellnum_{\givenclassifier}]$ where $\cellnum_{\givenclassifier} \triangleq |\signalspace_{\givenclassifier}|$.
For each $\ell\in[\cellnum_{\givenclassifier}]$, define the signal vector $\signalvec_{\givenclassifier,\ell}\in\reals_+^{\feanum}$ by
$\signalvec_{\givenclassifier,\ell}(i)
\triangleq
\kernal_{\givenclassifier}(\ell\mid \fea_i)$ for $i\in[\feanum]$.
The multiset $\signalbagraw(\givenclassifier)
\triangleq
\{\signalvec_{\givenclassifier,\ell}:\ell\in[\cellnum_{\givenclassifier}]\}$
is called the {\em raw bag-of-vectors} representation of predictor $\givenclassifier$.
\end{definition}

Since the signal kernel $\kernal_{\givenclassifier}(\cdot \mid \fea_i)$ is a probability distribution on $[\cellnum_{\givenclassifier}]$ for each $i\in[\feanum]$, the vectors in $\signalbagraw(\givenclassifier)$ satisfy
$\signalvec_{\givenclassifier,\ell}\in \reals_+^{\feanum}$ for all $\ell\in[\cellnum_{\givenclassifier}]$ and 
$\sum\nolimits_{\ell\in[\cellnum_{\givenclassifier}]}\signalvec_{\givenclassifier,\ell}
= \one$.
Conversely, any finite bag $\{\signalvec_1,\dots,\signalvec_{\cellnum}\}\subseteq \reals_+^{\feanum}$ satisfying $\sum_{\ell\in[\cellnum]}\signalvec_\ell=\one$
defines a finitely supported signal kernel: one simply takes signal space $[\cellnum]$ and sets $\kernal(\ell\mid \fea_i)\triangleq \signalvec_\ell(i)$
for $i\in[\feanum],\ \ell\in[\cellnum]$.
The coordinatewise nonnegativity and summation condition ensure that $\kernal(\cdot\mid \fea_i)$ is a probability distribution on $[\cellnum]$ for every $i\in[\feanum]$.
Thus, up to relabeling of signals, finitely supported signal kernels are in one-to-one correspondence with finite bags of nonnegative vectors whose coordinatewise sum is $\one$.

\wtedit{
The raw bag-of-vectors representation contains a redundancy: two nonzero vectors on the same positive ray encode the same informative signal, up to an independent random split.
Accordingly, we normalize out this redundancy.
\begin{definition}[Ray equivalence and reduced bag]
\label{def:reduced-bag}
For nonzero vectors $\signalvec,\signalvec\primed\in\reals_+^{\feanum}$, we write $\signalvec \rayeq \signalvec\primed$ if and only if $ \signalvec\primed = c\,\signalvec$ for some $c > 0$.
Let $\signalbagraw(\givenclassifier)=\{\signalvec_{\givenclassifier,\ell}\}_{\ell\in[\cellnum_{\givenclassifier}]}$.
Let $\mathcal{R}_1,\dots,\mathcal{R}_L$ be the distinct $\rayeq$-equivalence classes.
For each $t\in[L]$, define the merged vector $\bar{\signalvec}_{\givenclassifier,t}
\triangleq
\sum_{\signalvec_{\givenclassifier,\ell}\in \mathcal{R}_t}
\signalvec_{\givenclassifier,\ell}$.
The reduced bag-of-vectors representation of $\givenclassifier$ is $\signalbag(\givenclassifier)
\triangleq
\{\bar{\signalvec}_{\givenclassifier,t}:t\in[L]\}$.
\end{definition}}
If $\signalvec\primed=c\,\signalvec$ for some $c>0$, then $\feasubsetmean(\signalvec\primed)=\feasubsetmean(\signalvec)$.
Thus, all vectors in one ray class carry the same numeric prediction.
Moreover, passing from $\signalbagraw(\givenclassifier)$ to $\signalbag(\givenclassifier)$ only merges signals that differ by an independent random split, so this canonicalization does not change the informational-dominance content of the predictor.


From now on, whenever we write $\signalbag(\givenclassifier)$, we always mean the reduced bag-of-vectors representation.

\xhdr{Informational dominance in bag-of-vectors form}
With the above reformulation, we can also define informational dominance in bag-of-vectors form.
\begin{definition}[Informational dominance in bag form]
\label{def:partition-dom-bag}
Let $\signalbag(\givenclassifier)=\{\signalvec_{\givenclassifier,\ell}\}_{\ell\in[\cellnum_{\givenclassifier}]}$
and $\signalbag(\classifier)=\{\signalvec_{\classifier,r}\}_{r\in[\cellnum_{\classifier}]}$ be the reduced bags of two predictors.
We say that $\givenclassifier$ informationally-dominates $\classifier$, denoted by $\givenclassifier \partitiondom \classifier$, if and only if the following two conditions hold:
\begin{itemize}
    \item
    there exists a row-stochastic matrix $\kernalmtx=(\kernalmtx_{\ell, r})_{\ell\in[\cellnum_{\givenclassifier}],\,r\in[\cellnum_{\classifier}]}\in [0,1]^{\cellnum_{\givenclassifier}\times \cellnum_{\classifier}}$ such that, for every $r\in[\cellnum_{\classifier}]$, we have
    $ \signalvec_{\classifier,r}
    = \sum\nolimits_{\ell\in[\cellnum_{\givenclassifier}]}\kernalmtx_{\ell, r}\,\signalvec_{\givenclassifier,\ell}$.
    \item
    there does not exist any row-stochastic matrix $\kernalmtx=(\kernalmtx_{r, \ell})_{r\in[\cellnum_{\classifier}],\,\ell\in[\cellnum_{\givenclassifier}]}\in [0,1]^{\cellnum_{\classifier}\times \cellnum_{\givenclassifier}}$ such that, for every $\ell\in[\cellnum_{\givenclassifier}]$, we have $\signalvec_{\givenclassifier,\ell}
    =
    \sum\nolimits_{r\in[\cellnum_{\classifier}]}\kernalmtx_{r, \ell}\,\signalvec_{\classifier,r}$.
\end{itemize}
\end{definition}
Equivalently, if $\signalvecmtx_{\givenclassifier}\in\reals_+^{\cellnum_{\givenclassifier}\times \feanum}$ and $\signalvecmtx_{\classifier}\in\reals_+^{\cellnum_{\classifier}\times \feanum}$ denote the matrices whose rows are the vectors in $\signalbag(\givenclassifier)$ and $\signalbag(\classifier)$, respectively, then the first condition above is equivalent to $\signalvecmtx_{\classifier}=\kernalmtx^\top \signalvecmtx_{\givenclassifier}$.

\begin{lemma}
\label{lem:partition dom bag equiv}
\Cref{def:partition-dom-bag} is equivalent to the signal-kernel definition of informational dominance in \Cref{defn:partition dom}.
\end{lemma}

\begin{proof}[Proof of \Cref{lem:partition dom bag equiv}]
We write the signal spaces $\signalspace_{\givenclassifier}=[\cellnum_{\givenclassifier}]$ and $\signalspace_{\classifier}=[\cellnum_{\classifier}]$, and let
$\signalbag(\givenclassifier)=\{\signalvec_{\givenclassifier,\ell}\}_{\ell\in[\cellnum_{\givenclassifier}]}$ and
$\signalbag(\classifier)=\{\signalvec_{\classifier,r}\}_{r\in[\cellnum_{\classifier}]}$ be the corresponding reduced bags.

Assume first that $\givenclassifier \partitiondom \classifier$ in the sense of \Cref{defn:partition dom}.
Then there exists a stochastic kernel $\mapkernal:\signalspace_{\givenclassifier}\to \Delta(\signalspace_{\classifier})$ such that
\begin{align*}
    \kernal_{\classifier}(r\mid \fea_i)
    =
    \sum\nolimits_{\ell\in[\cellnum_{\givenclassifier}]}
    \mapkernal(r\mid \ell)\,\kernal_{\givenclassifier}(\ell\mid \fea_i)
    \quad
    \text{for all } i\in[\feanum],\ r\in[\cellnum_{\classifier}]~.
\end{align*}
Define the matrix $\kernalmtx$ with the entry  $\kernalmtx_{\ell, r}\triangleq \mapkernal(r\mid \ell)$.
Then $\kernalmtx=(\kernalmtx_{\ell, r})$ is row-stochastic, and for every $r\in[\cellnum_{\classifier}]$ and $i\in[\feanum]$,
\begin{align*}
    \signalvec_{\classifier,r}(i)
    =
    \kernal_{\classifier}(r\mid \fea_i)
    =
    \sum\nolimits_{\ell\in[\cellnum_{\givenclassifier}]}\kernalmtx_{\ell, r}\,\kernal_{\givenclassifier}(\ell\mid \fea_i)
    =
    \sum\nolimits_{\ell\in[\cellnum_{\givenclassifier}]}\kernalmtx_{\ell, r}\,\signalvec_{\givenclassifier,\ell}(i)~.
\end{align*}
Thus, the first condition in \Cref{def:partition-dom-bag} holds.
The second condition is translated identically.
Thus $\givenclassifier \partitiondom \classifier$ also holds in bag form.

Conversely, suppose \Cref{def:partition-dom-bag} holds.
Define a stochastic kernel $\mapkernal:\signalspace_{\givenclassifier}\to \Delta(\signalspace_{\classifier})$ by
$\mapkernal(r\mid \ell)\triangleq \kernalmtx_{\ell, r}$.
Then, for every $i\in[\feanum]$ and $r\in[\cellnum_{\classifier}]$,
\begin{align*}
    \kernal_{\classifier}(r\mid \fea_i)
    =
    \signalvec_{\classifier,r}(i)
    =
    \sum\nolimits_{\ell\in[\cellnum_{\givenclassifier}]}\kernalmtx_{\ell, r}\,\signalvec_{\givenclassifier,\ell}(i)
    =
    \sum\nolimits_{\ell\in[\cellnum_{\givenclassifier}]}\mapkernal(r\mid \ell)\,\kernal_{\givenclassifier}(\ell\mid \fea_i)~.
\end{align*}
So $\classifier$ is obtained from $\givenclassifier$ by a feature-independent garbling.
Again, the nonexistence of a reverse row-stochastic matrix is exactly the nonexistence of a reverse garbling kernel.
Thus, the two definitions are equivalent.
\end{proof}

\subsection{Constructible Predictors from Observable Linear Information}
\label{subsec:constructible and linear-infor}
In this section, we define the set of constructible predictors.
Fix the input predictors $(\givenclassifier_j)_{j\in[\classifiernum]}$, and let $(\signalbag(\givenclassifier_j))_{j\in[\classifiernum]}$ where 
$\signalbag(\givenclassifier_j)=\{\signalvec_{j,\ell}\}_{\ell\in[\cellnum_j]}$ denote their reduced bags.
Define $\totalcells \triangleq \sum\nolimits_{j\in[\classifiernum]}\cellnum_j$, and define the unknown vector
$\weightedbayes
\triangleq
(\feaprob_1\bayesprob_1,\dots,\feaprob_{\feanum}\bayesprob_{\feanum})^\top \in \reals^{\feanum}$.

\begin{definition}[Signal-feature matrix]
\label{def:signal-feature-matrix}
Define the signal-feature matrix $\mymatrix\in\reals_+^{\totalcells\times \feanum}$ whose rows are the signal vectors of the input predictors.
Equivalently, if row $(j,\ell)$ corresponds to the vector $\signalvec_{j,\ell}$, then
\begin{align*}
    \mymatrix_{(j,\ell),i}
    \triangleq
    \signalvec_{j,\ell}(i)
    \quad
    \text{for all } j\in[\classifiernum],\ \ell\in[\cellnum_j],\ i\in[\feanum].
\end{align*}
\aidelete{Let the vector $\weightedpredic\triangleq \mymatrix \weightedbayes\in\reals^{\totalcells}$
where its $(j,\ell)$-th component is
$\weightedpredic_{(j,\ell)}
=
\feasubsetlabel(\signalvec_{j,\ell})
=
\feasubsetmean(\signalvec_{j,\ell})\,\feasubsetprob(\signalvec_{j,\ell})$.\footnote{This scalar is determined by the observed inputs $(\feaprob,(\givenclassifier_j)_{j\in[\classifiernum]})$ since
$\feasubsetmean(\signalvec_{j,\ell})$ is the prediction attached to the signal $\signalvec_{j,\ell}$ and
$\feasubsetprob(\signalvec_{j,\ell})$ is computable from $\feaprob$ and $\signalvec_{j,\ell}$.}
Thus, the aggregator's information about the unknown vector $\weightedbayes$ is exactly the linear observation $\mymatrix\weightedbayes=\weightedpredic$.}
\aiedit{Let $\weightedpredic\in\reals^{\totalcells}$ denote the vector of observable signal-wise label masses, whose $(j,\ell)$-th component is defined by
$\weightedpredic_{(j,\ell)}
\triangleq
\feasubsetmean(\signalvec_{j,\ell})\,\feasubsetprob(\signalvec_{j,\ell})$.
Under signal-level calibration, the numerical prediction associated with each input signal equals its conditional label mean. Consequently, for every $j\in[\classifiernum]$ and $\ell\in[\cellnum_j]$,\aiedit{\footnote{\aiedit{We note that this construction relies on signal-level calibration. Under prediction-level calibration alone, the numerical prediction attached to an individual signal need not equal its conditional label mean when multiple signals share the same numerical prediction. Consequently, the individual signal-wise label masses might not be observable. 
We discuss this distinction and the corresponding prediction-level formulation in Appendix~\ref{apx:calibration}.}}}
\begin{align*}
    \weightedpredic_{(j,\ell)}
    =\feasubsetmean(\signalvec_{j,\ell})\,\feasubsetprob(\signalvec_{j,\ell})
    =\feasubsetlabel(\signalvec_{j,\ell})
    =\sum\nolimits_{i\in[\feanum]}\feaprob_i\bayesprob_i\signalvec_{j,\ell}(i)
    =\signalvec_{j,\ell}^{\top}\weightedbayes
    =(\mymatrix\weightedbayes)_{(j,\ell)}~.
\end{align*}
Therefore, the observable quantities collected in $\weightedpredic$ satisfy $\mymatrix\weightedbayes=\weightedpredic$. The matrix $\mymatrix$ is determined by the input signal structures, whereas $\weightedpredic$ is computed from the numerical predictions associated with the input signals and the known feature distribution $\feaprob$. 
}
\end{definition}
For the fixed inputs $(\feaprob,(\givenclassifier_j)_{j\in[\classifiernum]})$, the signal-wise label totals
$\feasubsetlabel(\signalvec_{j,\ell})$ are exactly the linear functionals of the unknown vector
$\weightedbayes=(\feaprob_1\bayesprob_1,\dots,\feaprob_{\feanum}\bayesprob_{\feanum})^\top$
that are directly determined by the observed predictors.
Collecting these quantities over all input signals yields a linear system
$\mymatrix\weightedbayes=\weightedpredic$;
we call the pair $(\mymatrix,\weightedpredic)$, or equivalently the equation
$\mymatrix\weightedbayes=\weightedpredic$, the {\em observable linear information}
induced by $(\feaprob,(\givenclassifier_j)_{j\in[\classifiernum]})$.

If predictor $\givenclassifier_j$ is deterministic, then for each feature $\fea_i$ there exists a unique signal in $\signalspace_{\givenclassifier_j}$ that is realized with probability one.
Define the cell
\begin{align*}
    \cell_{j,\ell}
    \triangleq
    \{\,i\in[\feanum]: \kernal_{\givenclassifier_j}(\ell\mid \fea_i)=1\,\},
    \quad
    \ell\in[\cellnum_j]~.
\end{align*}
Because $\givenclassifier_j$ is deterministic, the nonempty sets $\cell_{j,1},\dots,\cell_{j,\cellnum_j}$ form a partition of $[\feanum]$.
We call this partition the {\em signal partition} of $\givenclassifier_j$ and denote it by
$\partition_{\givenclassifier_j}=\{\cell_{j,\ell}\}_{\ell\in[\cellnum_j]}$.

\begin{remark}
\label{rem:det-special-case-M}
If predictor $\givenclassifier_j$ is deterministic with signal partition $\partition_{\givenclassifier_j}=\{\cell_{j,\ell}\}_{\ell\in[\cellnum_j]}$, then
$\signalvec_{j,\ell}=\subsetindi_{\cell_{j,\ell}}$ for every $\ell\in[\cellnum_j]$.
Thus, in the deterministic special case, \Cref{def:signal-feature-matrix} reduces to a matrix with $0/1$ rows.
\end{remark}

We identify the row space of $\mymatrix$ with a subspace of $\reals^{\feanum}$:
\begin{align*}
    \rowspace(\mymatrix)
    \triangleq
    \{\mymatrix^\top \alpha:\alpha\in\reals^{\totalcells}\}
    =
    \myspan\{\signalvec_{j,\ell}:j\in[\classifiernum],\ \ell\in[\cellnum_j]\}
    \subseteq \reals^{\feanum}.
\end{align*}

\begin{definition}[Observable vector]
\label{def:identifiable-vector}
A vector $\signalvec\in\reals^{\feanum}$ is called {\em \wtedit{observable}} if the value $\feasubsetlabel(\signalvec)=\signalvec^\top \weightedbayes$ is uniquely determined by the observed data, equivalently by $\weightedpredic=\mymatrix\weightedbayes$.
\end{definition}

\begin{proposition}
\label{prop:identifiable-vector}
A vector $\signalvec\in\reals^{\feanum}$ is {\em observable} if and only if $\signalvec\in \rowspace(\mymatrix)$.
\end{proposition}

\begin{proof}
Assume first that $\signalvec\in\rowspace(\mymatrix)$.
Then $\signalvec=\mymatrix^\top \alpha$ for some $\alpha\in\reals^{\totalcells}$, and therefore
\begin{align*}
    \feasubsetlabel(\signalvec)
    =
    \signalvec^\top \weightedbayes
    =
    \alpha^\top \mymatrix \weightedbayes
    =
    \alpha^\top \weightedpredic~,
\end{align*}
which is determined by $\weightedpredic$.

Conversely, suppose $\signalvec\notin \rowspace(\mymatrix)$.
Since $\rowspace(\mymatrix)=(\ker \mymatrix)^\perp$, there exists $\vz\in \ker \mymatrix$ such that $\signalvec^\top \vz\neq 0$.
Fix any $\weightedbayes^0\in\reals^{\feanum}$ satisfying $\mymatrix \weightedbayes^0=\weightedpredic$.
Then, for every scalar $\scalescalar\in\reals$,
\begin{align*}
    \mymatrix(\weightedbayes^0+\scalescalar\vz)
    =
    \mymatrix \weightedbayes^0 + \scalescalar\,\mymatrix \vz
    =
    \weightedpredic~.
\end{align*}
Thus, $\weightedbayes^0$ and $\weightedbayes^0+\scalescalar\vz$ yield the same observation $\weightedpredic$.
However,
\begin{align*}
    \signalvec^\top(\weightedbayes^0+\scalescalar\vz)
    =
    \signalvec^\top \weightedbayes^0 + \scalescalar\,\signalvec^\top \vz
\end{align*}
depends nontrivially on $\scalescalar$ because $\signalvec^\top \vz\neq 0$.
Thus $\feasubsetlabel(\signalvec)$ is not uniquely determined by $\weightedpredic$.
\end{proof}

\begin{corollary}
\label{cor:identifiable-subset}
We say a nonempty subset $\feasubset\subseteq[\feanum]$ is {\em observable} if and only if $\subsetindi_{\feasubset}\in \rowspace(\mymatrix)$.
\end{corollary}

\begin{proof}
Apply \Cref{prop:identifiable-vector} to the vector $\subsetindi_{\feasubset}$.
\end{proof}

\xhdr{Constructible predictors from observable information}
We are now ready to formally define constructible predictors.
\begin{definition}[Constructible predictors]
\label{def:constructible-predictor}
A finitely supported predictor $\classifier$ is said to be {\em constructible} from observable information if its reduced bag can be written as 
$\signalbag(\classifier)
=
\{\signalvec_{\classifier,1},\dots,\signalvec_{\classifier,\cellnum_{\classifier}}\}$
for some finite family of vectors satisfying:
\begin{itemize}
    \item[(i)]
    $\signalvec_{\classifier,r}\in \rowspace(\mymatrix)\cap \reals_+^{\feanum}\setminus\{\vzero\}$ for every $r\in[\cellnum_{\classifier}]$;
    \item[(ii)]
    $\sum\nolimits_{r\in[\cellnum_{\classifier}]}\signalvec_{\classifier,r}=\one$ coordinatewise.
\end{itemize}
The prediction attached to the signal $\signalvec_{\classifier,r}$ is $\feasubsetmean(\signalvec_{\classifier,r})$.
We write $\classifierfam$ for the set of all such predictors.
\end{definition}
Condition (i) ensures that each signal of $\classifier$ is identifiable from the observable linear information.
Condition (ii) ensures that the family $\{\signalvec_{\classifier,r}\}_{r\in[\cellnum_{\classifier}]}$ defines a valid stochastic kernel via
$\kernal_{\classifier}(r\mid \fea_i)\triangleq \signalvec_{\classifier,r}(i)$.

\begin{remark}
\label{rem:formal-constructibility}
\Cref{def:constructible-predictor} is the formal notion of constructibility used in the remainder of the paper.
In particular, whenever we write $\classifierfam$, we mean the family of predictors constructible from observable linear information in the sense of \Cref{def:constructible-predictor}.
\end{remark}

\begin{corollary}[Constructible deterministic predictors]
\label{cor:constructible-det}
Let $\classifier$ be a deterministic predictor, and let $\partition_{\classifier}=\{\cell_1,\dots,\cell_{\cellnum_{\classifier}}\}$ be its signal partition.
Then $\classifier\in \classifierfamdet$ if and only if $\subsetindi_{\cell_r}\in \rowspace(\mymatrix)$ for every $r\in[\cellnum_{\classifier}]$.
In that case, the prediction on the signal cell $\cell_r$ is $\feasubsetmean(\cell_r)$.
\end{corollary}

\begin{proof}[Proof of \Cref{cor:constructible-det}]
If predictor $\classifier$ is deterministic, then each feature produces exactly one signal with probability one.
Thus, the reduced bag of predictor $\classifier$ is
\begin{align*}
    \signalbag(\classifier)
    =
    \{\subsetindi_{\cell_1},\dots,\subsetindi_{\cell_{\cellnum_{\classifier}}}\}~,
\end{align*}
where $\partition_{\classifier}=\{\cell_1,\dots,\cell_{\cellnum_{\classifier}}\}$ is a partition of $[\feanum]$.
Thus, \Cref{def:constructible-predictor} is equivalent to the condition that $\subsetindi_{\cell_r}\in\rowspace(\mymatrix)$ for every $r\in[\cellnum_{\classifier}]$, and the attached prediction is
$\feasubsetmean(\subsetindi_{\cell_r})$.
\end{proof}

\begin{remark}
\label{rem:det-signal-partition}
For deterministic predictors under informational dominance, the primitive objects are the signal cells in $\partition_{\classifier}$.
Distinct signal cells are \emph{not} merged merely because they display the same numeric prediction.
This is the interpretation that will be used in \Cref{sec:det}.
\end{remark}

\section{Hardness and Efficient Algorithms for Deterministic Predictors}
\label{sec:det}

\newcommand{\partdecision}{\textsc{ifRefinable}}
\newcommand{\xthree}{\textsc{X3C3}}
\newcommand{\feaspaceone}{\feaspace_1}
\newcommand{\feaspacetwo}{\feaspace_2}

\newcommand{\xthreeset}{\mathcal{U}}
\newcommand{\element}{u}
\newcommand{\itemone}{v}
\newcommand{\itemtwo}{w}
\newcommand{\xthreecollec}{\mathcal{C}}
\newcommand{\xthreesubset}{C}
\newcommand{\xthreeincidencemtx}{\mathbf{X}}

\newcommand{\overlapgraph}{G}
\newcommand{\vertexset}{\mathcal{I}}
\newcommand{\vertexsetsize}{I}
\newcommand{\edgeset}{\mathcal{E}}
\newcommand{\edgesize}{E}
\newcommand{\connectedcomp}{K}
\newcommand{\inputcelllabel}{L}
\newcommand{\refinedpartition}{\widehat{\partition}}

In this section, we study the search problem {\predicrefinelink} restricted to deterministic inputs and deterministic outputs; that is, the input predictors $(\givenclassifier_j)_{j\in[\classifiernum]}$ and the output predictors $(\classifier_j)_{j\in[\classifiernum]}$ are all required to be deterministic predictors.

Fix a feature distribution $\feaprob=(\feaprob_i)_{i\in[\feanum]}$ and deterministic input predictors $(\givenclassifier_j)_{j\in[\classifiernum]}$.
Recall that $\classifierfamdet\subseteq \classifierfam$ denotes the set of constructible deterministic predictors, and $\partitiondomparedet\subseteq \classifierfamdet$ denotes the set of constructible deterministic predictors that are undominated by any other constructible deterministic predictors.
By \Cref{cor:constructible-det}, a deterministic predictor $\classifier$ belongs to $\classifierfamdet$ if and only if $\subsetindi_{\cell}\in \rowspace(\mymatrix)$ for every signal cell $\cell\in \partition_{\classifier}$.
In that case, the prediction on cell $\cell$ is the calibrated value $\feasubsetmean(\subsetindi_{\cell})$.
By \Cref{rem:det-signal-partition}, the primitive objects in this section are the signal cells in $\partition_{\classifier}$.
With this simplification in place, informational dominance for deterministic predictors can be characterized entirely in terms of the refinement relation between their signal partitions.
\begin{lemma}[Dominance for deterministic predictors]
\label{lem:det-partition-dom}
Let $\classifier$ and $\classifiernew$ be deterministic predictors, and let
$\partition_{\classifier}$ and $\partition_{\classifiernew}$ denote their signal partitions.
Then we have $\classifier \partitiondom \classifiernew$ if and only if $\partition_{\classifier}$ {\em strictly refines} $\partition_{\classifiernew}$, that is, 
for every cell $\cell\in \partition_{\classifier}$, there exists a cell $\cell'\in \partition_{\classifiernew}$
such that $\cell\subseteq \cell'$, and moreover $\partition_{\classifier}\neq \partition_{\classifiernew}$.
\end{lemma}

\subsection{$\mathsf{NP}$-Hardness for Solving {\predicrefinelink}}
\label{subsec:hardness}

\newcommand{\partdecisionlink}{\hyperref[prob:partdecision]{\partdecision}}

To establish the hardness of the search problem {\predicrefinelink}, we first consider the following decision problem {\partdecisionlink} as a sub-problem toward solving {\predicrefinelink}.
\phantomsection
\label{prob:partdecision}
\mybox{$\partdecision$: \\
\textbf{Input:} a feature distribution $\feaprob=(\feaprob_i)_{i\in[\feanum]}$ and $\classifiernum$ deterministic predictors $(\givenclassifier_j)_{j\in[\classifiernum]}$. \\
\textbf{Question:} does there exist a predictor $\classifier\in \partitiondomparedet$ such that $\classifier \partitiondom \givenclassifier_1$?
}

\begin{theorem}
\label{thm:hardness partdecision}
The decision problem {\partdecisionlink} is $\mathsf{NP}$-hard.
\end{theorem}

\begin{corollary}
\label{cor:hardness predicrefine}
The search problem {\predicrefinelink}, restricted to deterministic input predictors and deterministic output predictors, is $\mathsf{NP}$-hard.
\end{corollary}

At a high level, the proof of \Cref{thm:hardness partdecision} proceeds by a reduction from the following $\mathsf{NP}$-complete problem, which we denote by {\xthree} (see, e.g., \citealp{GJ-02}).

\begin{tcolorbox}[colback=white,colframe=black,boxrule=0.8pt,arc=0pt]
Exact $3$-cover with multiplicity $3$ ({\xthree}): \\
\textbf{Input:} a set $\xthreeset=\{\element_1,\ldots,\element_\cellnum\}$ and a collection $\xthreecollec=\{\xthreesubset_1,\ldots,\xthreesubset_d\}$ of $3$-element subsets of $\xthreeset$ such that every element $\element_r\in \xthreeset$ belongs to exactly three sets, and $m$ is a multiple of $3$. \\
\textbf{Question:} does there exist an exact cover of $\xthreeset$?
\end{tcolorbox}
The problem {\xthree} encodes the incidence matrix $\xthreeincidencemtx\in\{0,1\}^{\cellnum\times d}$ into the observable linear space $\rowspace(\mymatrix)$. 
The construction creates two feature blocks $\feaspace = \feaspaceone\sqcup \feaspacetwo$, takes the first predictor to have partition $\partition_{\givenclassifier_1}=\{\feaspaceone,\feaspacetwo\}$, and for each set $\xthreecollec_j$ introduces a two-cell predictor built from $R_j=\{\itemone_j\}\cup\{\itemtwo_r:\element_r\in \xthreecollec_j\}$. 
As a result, by \Cref{lem:det-partition-dom} and \Cref{cor:constructible-det}, a strict deterministic improvement of $\givenclassifier_1$ is equivalent to splitting one of the two cells $\feaspaceone$ or $\feaspacetwo$ by a nonempty proper observable subset $\feasubset$, namely a subset satisfying $\subsetindi_\feasubset\in \rowspace(\mymatrix)$. 
The key structural step is then to characterize exactly when such subsets exist: \Cref{lem:cell two} shows that $\feaspacetwo$ admits no nontrivial observable subset, 
while \Cref{lem:cell one} shows that for $\feasubset\subseteq \feaspaceone$ with indicator vector $\vz\in\{0,1\}^d$, the condition $\subsetindi_\feasubset\in \rowspace(\mymatrix)$ is equivalent to $\xthreeincidencemtx \vz=\scalescalar \subsetindi_\cellnum$ for some scalar $\scalescalar$. 
The four possible values $\scalescalar\in\{0,1,2,3\}$ then separate the two trivial cases from the two exact-cover cases, and \Cref{lem:reduce to x3c3} packages this into the equivalence between an exact cover of the {\xthree} instance, the existence of some constructible predictor $\classifier\in \classifierfamdet$ with $\classifier\partitiondom \givenclassifier_1$, and the existence of an undominated one $\classifier^\star\in \partitiondomparedet$ with $\classifier^\star\partitiondom \givenclassifier_1$.

The proof below makes this reduction precise; once \Cref{thm:hardness partdecision} is established, \Cref{cor:hardness predicrefine} follows immediately from the specification of {\predicrefinelink} (its proof is deferred to \Cref{apx:proofs det}).

\begin{proof}[Proof of \Cref{thm:hardness partdecision}]
Let $\xthreeincidencemtx\in\{0,1\}^{\cellnum\times d}$ be the incidence matrix of the {\xthree} instance, where $\xthreeincidencemtx_{r,j}=1$ if and only if $\element_r\in \xthreesubset_j$.
Since every column and every row of $\xthreeincidencemtx$ has sum $3$, double-counting incidences gives
\begin{align*}
    \sum\nolimits_{r\in[\cellnum]}\sum\nolimits_{j\in[d]}\xthreeincidencemtx_{r,j}
    =
    3d
    =
    3\cellnum~,
\end{align*}
and hence $d=\cellnum$.

\xhdr{Feature distribution and initial predictors}
We construct an instance of {\partdecisionlink} as follows.
Create $\feanum=d+\cellnum=2\cellnum$ features and partition the feature space into two blocks
\begin{align*}
    \feaspace
    =
    \feaspaceone \sqcup \feaspacetwo,
    \quad
    \feaspaceone\triangleq \{\itemone_1,\ldots,\itemone_d\},
    \quad
    \feaspacetwo\triangleq \{\itemtwo_1,\ldots,\itemtwo_\cellnum\}~.
\end{align*}
Let the feature distribution be uniform: $\feaprob_i=\sfrac{1}{\feanum}$ for every $i\in[\feanum]$.
We define $\classifiernum=d+1=\cellnum+1$ deterministic input predictors $(\givenclassifier_j)_{j\in[\classifiernum]}$.

\begin{itemize}
    \item
    Predictor $\givenclassifier_1$ has signal partition $\partition_{\givenclassifier_1}
    =
    \{\feaspaceone,\feaspacetwo\}$,
    and predicts $\feasubsetmean(\subsetindi_{\cell})$ on each signal cell $\cell\in \partition_{\givenclassifier_1}$.

    \item
    For each $j\in[d]$, let $R_j
    \triangleq
    \{\itemone_j\}\cup \{\,\itemtwo_r:\ \element_r\in \xthreesubset_j\,\}$, and let the predictor $\givenclassifier_{j+1}$ have signal partition
    $\partition_{\givenclassifier_{j+1}}
    =
    \{R_j, R_j^c\}$
    and predict $\feasubsetmean(\subsetindi_{\cell})$ on each signal cell $\cell\in \partition_{\givenclassifier_{j+1}}$.
\end{itemize}
Finally, enumerate the $\feanum$ features arbitrarily as $\fea_1,\ldots,\fea_\feanum$, and let $\bayesprob_i
\gets
(\feanum+1)^{-i}$ be the Bayes probability
for the $i$-th feature.

\xhdr{Row-space description}
Let $\mymatrix$ be the signal-feature matrix from \Cref{def:signal-feature-matrix} for the constructed input predictors.
Since all input predictors are deterministic, \Cref{rem:det-special-case-M} implies that the rows of $\mymatrix$ are exactly the indicator vectors of the initial signal cells.
Because each signal partition $\partition_{\givenclassifier_{j+1}}$ has signal cells $R_j$ and $R_j^c$, and $\subsetindi_{R_j^c}
=
\subsetindi_{\feaspaceone}
+
\subsetindi_{\feaspacetwo}
-
\subsetindi_{R_j}$,
we obtain
\begin{align}
\label{eq:row matrix contruction}
    \rowspace(\mymatrix)
    =
    \myspan
    \left\{
        \subsetindi_{\feaspaceone},
        \subsetindi_{\feaspacetwo},
        \subsetindi_{R_1},
        \ldots,
        \subsetindi_{R_d}
    \right\}~.
\end{align}
We now characterize the observable subsets inside $\feaspaceone$ and $\feaspacetwo$.

\begin{lemma}
\label{lem:cell one}
Let $\feasubset\subseteq \feaspaceone$, and let $\vz\in\{0,1\}^d$ be the indicator vector of $\feasubset$, namely $\vz_j=1$ if and only if $\itemone_j\in \feasubset$.
Then we have $\subsetindi_{\feasubset}\in \rowspace(\mymatrix)$ if and only if $\xthreeincidencemtx \vz = \scalescalar \one_\cellnum$ 
for some scalar $\scalescalar\in\reals$.
\end{lemma}

\begin{proof}
Assume first that $\subsetindi_{\feasubset}\in \rowspace(\mymatrix)$.
By Eqn.~\eqref{eq:row matrix contruction}, there exist scalars $a,b,\alpha_1,\ldots,\alpha_d$ such that
\begin{align}
\label{eq:subset identify indi vec}
    \subsetindi_{\feasubset}
    =
    a\subsetindi_{\feaspaceone}
    +
    b\subsetindi_{\feaspacetwo}
    +
    \sum\nolimits_{j\in[d]}\alpha_j \subsetindi_{R_j}~.
\end{align}
Evaluating Eqn.~\eqref{eq:subset identify indi vec} at the feature $\itemone_j$ gives
\begin{align}
\label{eq:feature \itemone_j}
    \vz_j
    =
    a+\alpha_j,
    \quad\text{so}\quad
    \alpha_j = \vz_j-a~.
\end{align}
Now evaluate Eqn.~\eqref{eq:subset identify indi vec} at the feature $\itemtwo_r$.
Since $\itemtwo_r\notin \feaspaceone$, $\itemtwo_r\in \feaspacetwo$, and $\itemtwo_r\in R_j$ if and only if $\element_r\in \xthreesubset_j$, we obtain
\begin{align*}
    0
    =
    b+\sum\nolimits_{j\in[d]}\xthreeincidencemtx_{r,j}\alpha_j~.
\end{align*}
Substituting Eqn.~\eqref{eq:feature \itemone_j} yields
\begin{align*}
    0
    =
    b+\sum\nolimits_{j\in[d]}\xthreeincidencemtx_{r,j}(\vz_j-a) 
    =
    b+(\xthreeincidencemtx\vz)_r-a\sum\nolimits_{j\in[d]}\xthreeincidencemtx_{r,j}~.
\end{align*}
Each row of $\xthreeincidencemtx$ has sum $3$, so $(\xthreeincidencemtx\vz)_r=3a-b$ for every $r\in[\cellnum]$.
Hence $\xthreeincidencemtx\vz=(3a-b)\one_\cellnum$.

Conversely, suppose that $\xthreeincidencemtx \vz=\scalescalar \one_\cellnum$ for some scalar $\scalescalar\in\reals$.
Then
\begin{align*}
    \sum\nolimits_{j\in[d]} \vz_j \subsetindi_{R_j}
    -
    \scalescalar \subsetindi_{\feaspacetwo}
    \in \rowspace(\mymatrix)~.
\end{align*}
At the coordinate corresponding to $\itemone_j$, this vector equals $\vz_j$.
At the coordinate corresponding to $\itemtwo_r$, it equals $(\xthreeincidencemtx\vz)_r-\scalescalar=0$.
Thus the vector is exactly $\subsetindi_{\feasubset}$, so $\subsetindi_{\feasubset}\in \rowspace(\mymatrix)$.
\end{proof}

\begin{lemma}
\label{lem:cell two}
If $\feasubset\subseteq \feaspacetwo$ and $\subsetindi_{\feasubset}\in \rowspace(\mymatrix)$, then either $\feasubset=\emptyset$ or $\feasubset=\feaspacetwo$.
\end{lemma}

\begin{proof}
We write
\begin{align}
\label{eq:subset identify indi vec 2}
    \subsetindi_{\feasubset}
    =
    a\subsetindi_{\feaspaceone}
    +
    b\subsetindi_{\feaspacetwo}
    +
    \sum\nolimits_{j\in[d]} \alpha_j \subsetindi_{R_j}
\end{align}
for some scalars $a,b,\alpha_1,\ldots,\alpha_d$.
Evaluating Eqn.~\eqref{eq:subset identify indi vec 2} at $\itemone_j$ yields $0=a+\alpha_j$ for every $j$, because $\feasubset\subseteq \feaspacetwo$.
Thus $\alpha_j=-a$ for every $j$.
Evaluating Eqn.~\eqref{eq:subset identify indi vec 2} at $\itemtwo_r$ then gives
\begin{align*}
    (\subsetindi_{\feasubset})_{\itemtwo_r}
    =
    b+\sum\nolimits_{j\in[d]}\xthreeincidencemtx_{r,j}\alpha_j 
    =
    b-a\sum\nolimits_{j\in[d]}\xthreeincidencemtx_{r,j}
    =
    b-3a~.
\end{align*}
Thus, every coordinate of $\subsetindi_{\feasubset}$ on $\feaspacetwo$ has the same value.
Since $\subsetindi_{\feasubset}$ is a $\{0,1\}$-vector, either all these coordinates are $0$ or all are $1$.
Therefore $\feasubset=\emptyset$ or $\feasubset=\feaspacetwo$.
\end{proof}

\begin{lemma}
\label{lem:reduce to x3c3}
The following statements are equivalent:
\begin{enumerate}
    \item the {\xthree} instance admits an exact cover;
    \item there exists a predictor $\classifier\in\classifierfamdet$ such that $\classifier \partitiondom \givenclassifier_1$;
    \item there exists a predictor $\classifier^\star\in \partitiondomparedet$ such that $\classifier^\star \partitiondom \givenclassifier_1$.
\end{enumerate}
\end{lemma}

\begin{proof}
We first prove the equivalence between (1) and (2).

Suppose that there exists a predictor $\classifier\in \classifierfamdet$ such that $\classifier \partitiondom \givenclassifier_1$.
By the deterministic special case of \wtdelete{ partition} dominance, the signal partition $\partition_{\classifier}$ strictly refines $\partition_{\givenclassifier_1}=\{\feaspaceone,\feaspacetwo\}$.
Therefore at least one of the two cells $\feaspaceone$ or $\feaspacetwo$ contains a nonempty proper subset that is a signal cell of $\classifier$.
Because $\classifier\in\classifierfamdet$, every signal cell of $\classifier$ is observable by \Cref{cor:constructible-det}.
Thus at least one of $\feaspaceone$ or $\feaspacetwo$ contains a nonempty proper observable subset.

Conversely, suppose that one of the two cells in $\partition_{\givenclassifier_1}$ contains a nonempty proper observable subset.
That is, suppose there exist $\cell\in\{\feaspaceone,\feaspacetwo\}$ and a nonempty proper subset $\feasubset\subsetneq \cell$ such that $\subsetindi_{\feasubset}\in \rowspace(\mymatrix)$.
Since $\cell$ is itself a signal cell of $\givenclassifier_1$, we also have $\subsetindi_{\cell}\in \rowspace(\mymatrix)$.
Thus, we have
\begin{align*}
    \subsetindi_{\cell\setminus \feasubset}
    =
    \subsetindi_{\cell}-\subsetindi_{\feasubset}
    \in \rowspace(\mymatrix)~.
\end{align*}
Thus the partition obtained from $\{\feaspaceone,\feaspacetwo\}$ by replacing $\cell$ with the two cells $\feasubset$ and $\cell\setminus\feasubset$ has all signal-cell indicators in $\rowspace(\mymatrix)$.
By \Cref{cor:constructible-det}, this partition defines a predictor $\classifier\in \classifierfamdet$.
Its signal partition strictly refines $\partition_{\givenclassifier_1}$, so the deterministic special case of \wtdelete{ partition} dominance implies $\classifier \partitiondom \givenclassifier_1$.

We have shown that statement (2) is equivalent to the existence of a nonempty proper observable subset inside $\feaspaceone$ or $\feaspacetwo$.
By \Cref{lem:cell two}, the cell $\feaspacetwo$ has no nonempty proper observable subset.
Therefore statement (2) holds if and only if there exists a nonempty proper subset $\feasubset\subsetneq \feaspaceone$ such that $\subsetindi_{\feasubset}\in \rowspace(\mymatrix)$.

Now apply \Cref{lem:cell one}.
Statement (2) is equivalent to the existence of a vector $\vz\in\{0,1\}^d$ that is neither $\vzero$ nor $\one_d$ and satisfies
$\xthreeincidencemtx \vz = \scalescalar \one_\cellnum$
for some scalar $\scalescalar\in\reals$.
Because each row of $\xthreeincidencemtx$ has exactly three $1$'s and $\vz\in\{0,1\}^d$, every entry of $\xthreeincidencemtx \vz$ belongs to $\{0,1,2,3\}$.
Thus, we have $\scalescalar\in\{0,1,2,3\}$.
We now analyze the following four cases.
\begin{itemize}
    \item
    \emph{Case $\scalescalar=0$.}
    Then $\xthreeincidencemtx \vz=\vzero$.
    Multiplying on the left by $\one_\cellnum^\top$ gives
    \begin{align*}
        0
        =
        \one_\cellnum^\top \xthreeincidencemtx \vz
        =
        3\,\one_d^\top \vz~,
    \end{align*}
    because every column of $\xthreeincidencemtx$ has sum $3$.
    Hence $\one_d^\top \vz=0$, so $\vz=\vzero$, contradicting the fact that $\feasubset$ is nonempty.

    \item
    \emph{Case $\scalescalar=3$.}
    Then $\xthreeincidencemtx \vz=3\one_\cellnum$.
    Multiplying on the left by $\one_\cellnum^\top$ gives
    \begin{align*}
        3m
        =
        \one_\cellnum^\top \xthreeincidencemtx \vz
        =
        3\,\one_d^\top \vz~.
    \end{align*}
    Since $d=m$, we obtain $\one_d^\top \vz=d$.
    As $\vz\in\{0,1\}^d$, this implies $\vz=\one_d$, contradicting the fact that $\feasubset$ is a proper subset of $\feaspaceone$.

    \item
    \emph{Case $\scalescalar=1$.}
    Then $\xthreeincidencemtx \vz=\one_\cellnum$.
    Equivalently, every element of $\xthreeset$ belongs to exactly one selected set.
    Hence the selected family $\{\xthreesubset_j:\ \vz_j=1\}$
    is an exact cover of $\xthreeset$.

    \item
    \emph{Case $\scalescalar=2$.}
    Then
    \begin{align*}
        \xthreeincidencemtx(\one_d-\vz)
        =
        \xthreeincidencemtx\one_d-\xthreeincidencemtx\vz
        =
        3\one_\cellnum-2\one_\cellnum
        =
        \one_\cellnum~.
    \end{align*}
    Thus, the complementary family $\{\xthreesubset_j:\ \vz_j=0\}$
    is an exact cover of $\xthreeset$.
\end{itemize}

We have thus shown that (1) and (2) are equivalent.

It remains to prove the equivalence between (2) and (3).
The implication $(3)\Rightarrow(2)$ is immediate.
For the converse, assume that there exists $\classifier\in\classifierfamdet$ such that $\classifier \partitiondom \givenclassifier_1$.
Define
\begin{align*}
    \classifierfamdet_\classifier
    \triangleq
    \left\{
        \classifiernew\in\classifierfamdet:\ \partition_\classifiernew \text{ refines } \partition_{\classifier}
    \right\}~.
\end{align*}
By \Cref{cor:constructible-det}, every predictor in $\classifierfamdet$ is uniquely determined by its signal partition, and there are only finitely many set partitions of $[\feanum]$.
Hence $\classifierfamdet_\classifier$ is finite.
Choose $\classifier^\star\in\classifierfamdet_\classifier$ with the maximum possible number of signal cells.

We claim that $\classifier^\star\in\partitiondomparedet$.
If not, there exists some $\classifiernew\in\classifierfamdet$ such that $\classifiernew \partitiondom \classifier^\star$.
By the deterministic special case of \wtdelete{ partition} dominance, $\partition_{\classifiernew}$ strictly refines $\partition_{\classifier^\star}$.
Since $\partition_{\classifier^\star}$ refines $\partition_{\classifier}$, it follows that $\partition_{\classifiernew}$ also refines $\partition_{\classifier}$, so $\classifiernew\in\classifierfamdet_\classifier$.
But strict refinement increases the number of signal cells, contradicting the maximal choice of $\classifier^\star$.
Thus $\classifier^\star\in\partitiondomparedet$.

Finally, since $\partition_{\classifier^\star}$ refines $\partition_{\classifier}$ and $\partition_{\classifier}$ strictly refines $\partition_{\givenclassifier_1}$, the partition $\partition_{\classifier^\star}$ also strictly refines $\partition_{\givenclassifier_1}$.
Thus, again by the deterministic special case of \wtdelete{ partition} dominance, we have $\classifier^\star \partitiondom \givenclassifier_1$.
This proves $(2)\Rightarrow(3)$, and completes the proof.
\end{proof}

It remains to verify that the reduction has polynomial output size.
The combinatorial part is immediate: $\feanum=d+m=2m$, $\classifiernum=d+1=m+1$, and every input predictor has exactly two signal cells.
It remains to bound the encoding length of the numerical predictions of the input predictors.
Fix any initial signal cell $\cell$.
Because $\feaprob_i=\sfrac{1}{\feanum}$ for every $i$,
\begin{align*}
    \feasubsetmean(\subsetindi_{\cell})
    =
    \frac{\feasubsetlabel(\subsetindi_{\cell})}{\feaprob(\subsetindi_{\cell})}
    =
    \frac{\sum\nolimits_{i\in\cell}\feaprob_i \bayesprob_i}{\sum\nolimits_{i\in\cell}\feaprob_i} 
    =
    \frac{1}{|\cell|}
    \sum\nolimits_{i\in\cell} (\feanum+1)^{-i}
    =
    \frac{1}{|\cell|(\feanum+1)^\feanum}
    \sum\nolimits_{i\in\cell} (\feanum+1)^{\feanum-i}~.
\end{align*}
Thus $\feasubsetmean(\subsetindi_{\cell})$ is a rational number whose numerator and denominator are both at most $\feanum(\feanum+1)^\feanum$.
Thus, each such value has bit-length $O(\feanum\log \feanum)$.
Therefore the whole constructed instance of {\partdecisionlink} has polynomial size.

By \Cref{lem:reduce to x3c3}, the {\xthree} instance is a YES-instance if and only if the constructed instance of {\partdecisionlink} is a YES-instance.
Therefore {\xthree} reduces in polynomial time to {\partdecisionlink}.
Since {\xthree} is $\mathsf{NP}$-complete, the problem {\partdecisionlink} is $\mathsf{NP}$-hard.
\end{proof}

\subsection{Efficient Algorithm for Solving {\predicrefinelink} with $\classifiernum = 2$}
\label{subsec:efficient}

In this subsection, we show that there exists a polynomial-time algorithm that can solve the deterministic version of  problem {\predicrefinelink} when the number of input predictors $\classifiernum = 2$.

\begin{theorem}
\label{thm:det-k2-polytime}
Assume $\classifiernum=2$.
Then there exists a polynomial-time algorithm (Algorithms~\ref{alg:det-k2-finest} and~\ref{alg:det-k2-solve}) that solves the deterministic version of {\predicrefinelink}.\wtdelete{ under informational dominance}
\end{theorem}

We write $\partition_{\givenclassifier_1}
=
\{\cell_{1,\ell}\}_{\ell\in[\cellnum_1]},
\partition_{\givenclassifier_2}
=
\{\cell_{2,r}\}_{r\in[\cellnum_2]}$
for the signal partitions of the two deterministic input predictors $\givenclassifier_1$ and $\givenclassifier_2$.
By \Cref{rem:det-special-case-M}, the rows of the signal-feature matrix  $\mymatrix$ are exactly the indicator vectors $\subsetindi_{\cell_{1,\ell}}$ for $\ell\in[\cellnum_1]$ and 
$\subsetindi_{\cell_{2,r}}$ for $r\in[\cellnum_2]$.
For notational simplicity, for every input signal cell $\cell\in \partition_{\givenclassifier_1}\cup \partition_{\givenclassifier_2}$, define its observable cell label mass by $\inputcelllabel(\cell)
\triangleq
\feasubsetlabel(\subsetindi_{\cell})
=
\feasubsetmean(\subsetindi_{\cell})\,\feaprob(\subsetindi_{\cell})$
which is determined by the observed input predictor and the feature distribution.

Define the overlap graph $\overlapgraph$ as the bipartite graph with vertex set
\begin{align*}
    \vertexset(\overlapgraph)
    \triangleq
    \partition_{\givenclassifier_1}\sqcup \partition_{\givenclassifier_2}~,
\end{align*}
and with an edge between $\cell_{1,\ell}$ and $\cell_{2,r}$ if and only if $\cell_{1,\ell}\cap \cell_{2,r}\neq \emptyset$ (see \Cref{fig:algo1-left} for an illustration).

At a high level, the proof of \Cref{thm:det-k2-polytime} exploits the fact that when $\classifiernum=2$, the observable structure inside a fixed input cell $\cell$ is completely controlled by the overlap graph $\overlapgraph$ between the signal partitions $\partition_{\givenclassifier_1}$ and $\partition_{\givenclassifier_2}$. 
For each cell $\cell$ on one side, we remove the vertex $\cell$ from the overlap graph $\overlapgraph$ and use the connected components $\connectedcomp$ of $\overlapgraph-\cell$ to decompose the cell $\cell$ into its $\connectedcomp$-branches $\cell_\connectedcomp$ (see \Cref{fig:algo1-right} for an illustration). 
\Cref{lem:det-k2-branch} shows that these branch pieces are exactly the right atomic objects for refinement: part (i) shows that they partition the cell $\cell$, part (ii) shows that each $\cell_\connectedcomp$ is observable and that its label mass can be computed from the observable cell label masses $\inputcelllabel(\cdot)$, and part (iii) shows that every observable subset $\feasubset\subseteq \cell$ must be a union of such branch pieces. 
Consequently, if we refine every cell $\cell\in \partition_{\givenclassifier_j}$ into all of its branch pieces, we obtain a canonical partition, which we denote by $\refinedpartition_j$, as the finest constructible deterministic refinement of $\partition_{\givenclassifier_j}$. 
\Cref{lem:det-partition-dom} then converts this structural statement into the desired dominance conclusion: if $\refinedpartition_j \neq \partition_{\givenclassifier_j}$, the associated predictor $\classifier_j^\star$ strictly informationally-dominates $\givenclassifier_j$; if $\refinedpartition_j=\partition_{\givenclassifier_j}$, then no constructible deterministic predictor can strictly informationally-dominate $\givenclassifier_j$. \Cref{alg:det-k2-finest,alg:det-k2-solve} simply implement this canonical refinement and test which of these two cases occurs (see \Cref{fig:overlap-graph-algorithm-one} for an example and illustration).

\begin{figure}[h]
\centering
\begin{subfigure}[t]{0.47\textwidth}
\centering
\begin{tikzpicture}[
    scale=0.95,
    every node/.style={transform shape},
    cellblue/.style={
        circle,
        draw=blue!70!black,
        fill=blue!50,
        thick,
        minimum size=9mm,
        inner sep=0pt,
        font=\large\bfseries\color{black}
    },
    cellorange/.style={
        circle,
        draw=orange!70!black,
        fill=orange!75,
        thick,
        minimum size=9mm,
        inner sep=0pt,
        font=\large\bfseries\color{black}
    },
    edge/.style={draw=black, thick},
    setlabel/.style={font=\small}
]

\node[cellblue] (A1) at (0,2.15) {\(\cell_1\)};
\node[cellblue] (A2) at (0,0.40) {\(\cell_2\)};
\node[cellblue] (A3) at (0,-1.35) {\(\cell_3\)};

\node[cellorange] (B1) at (3.2,2.15) {\(B_1\)};
\node[cellorange] (B2) at (3.2,0.40) {\(B_2\)};
\node[cellorange] (B3) at (3.2,-1.35) {\(B_3\)};

\draw[edge] (A1) -- (B1);
\draw[edge] (A1) -- (B2);
\draw[edge] (A1) -- (B3);
\draw[edge] (A2) -- (B1);
\draw[edge] (A3) -- (B3);

\node[setlabel, left=4mm of A1] {\(\{\fea_1,\fea_2,\fea_3,\fea_4\}\)};
\node[setlabel, left=4mm of A2] {\(\{\fea_5\}\)};
\node[setlabel, left=4mm of A3] {\(\{\fea_6\}\)};

\node[setlabel, right=4mm of B1] {\(\{\fea_1,\fea_5\}\)};
\node[setlabel, right=4mm of B2] {\(\{\fea_2,\fea_3\}\)};
\node[setlabel, right=4mm of B3] {\(\{\fea_4,\fea_6\}\)};


\end{tikzpicture}
\caption{Original overlap graph \(\overlapgraph\) (edge iff the two cells intersect).}
\label{fig:algo1-left}
\end{subfigure}
\hfill
\begin{subfigure}[t]{0.47\textwidth}
\centering
\begin{tikzpicture}[
    scale=0.95,
    every node/.style={transform shape},
    cellblue/.style={
        circle,
        draw=blue!70!black,
        fill=blue!50,
        thick,
        minimum size=9mm,
        inner sep=0pt,
        font=\large\bfseries\color{black}
    },
    cellorange/.style={
        circle,
        draw=orange!70!black,
        fill=orange!75,
        thick,
        minimum size=9mm,
        inner sep=0pt,
        font=\large\bfseries\color{black}
    },
    removed/.style={
        circle,
        draw=black,
        dashed,
        thick,
        minimum size=9mm,
        inner sep=0pt,
        font=\large\bfseries
    },
    edge/.style={draw=black, thick},
    ghostedge/.style={draw=black, thick, dashed},
    compbox/.style={
        draw=blue!80,
        dashed,
        very thick,
        rounded corners=12pt,
        inner sep=5pt
    },
    klabel/.style={font=\small, anchor=west}
]

\node[removed] (A1r) at (0,0) {\(\cell_1\)};

\node[cellblue] (A2r) at (2.0,1.65) {\(\cell_2\)};
\node[cellorange] (B1r) at (3.25,1.65) {\(B_1\)};
\draw[edge] (A2r) -- (B1r);

\node[cellorange] (B2r) at (3.25,0.00) {\(B_2\)};

\node[cellblue] (A3r) at (2.0,-1.65) {\(\cell_3\)};
\node[cellorange] (B3r) at (3.25,-1.65) {\(B_3\)};
\draw[edge] (A3r) -- (B3r);

\draw[ghostedge] (A1r) -- (B1r);
\draw[ghostedge] (A1r) -- (B2r);
\draw[ghostedge] (A1r) -- (B3r);

\node[compbox, fit=(A2r)(B1r)] {};
\node[compbox, fit=(B2r)] {};
\node[compbox, fit=(A3r)(B3r)] {};

\node[klabel] at (4.65,1.65) {\(\connectedcomp_1=\{\cell_2,B_1\}\)};
\node[klabel] at (4.65,0.00) {\(\connectedcomp_2=\{B_2\}\)};
\node[klabel] at (4.65,-1.65) {\(\connectedcomp_3=\{\cell_3,B_3\}\)};

\end{tikzpicture}
\caption{Connected components of \(\overlapgraph-\cell_1\).}
\label{fig:algo1-right}
\end{subfigure}
\caption{
Illustration of \Cref{alg:det-k2-finest}. 
In this example, the two deterministic predictors \(\givenclassifier_1\) and \(\givenclassifier_2\) induce signal partitions over the feature set. We write each signal cell by the feature indices it contains:
\(\partition_{\givenclassifier_1}=\{\cell_1,\cell_2,\cell_3\}\), where
\(\cell_1=\{\fea_1,\fea_2,\fea_3,\fea_4\}\),
\(\cell_2=\{\fea_5\}\), and
\(\cell_3=\{\fea_6\}\); and
\(\partition_{\givenclassifier_2}=\{B_1,B_2,B_3\}\), where
\(B_1=\{\fea_1,\fea_5\}\),
\(B_2=\{\fea_2,\fea_3\}\), and
\(B_3=\{\fea_4,\fea_6\}\).
\Cref{fig:algo1-left} shows the overlap graph \(\overlapgraph\), whose vertices are the signal cells of \(\givenclassifier_1\) and \(\givenclassifier_2\), and where an edge indicates a nonempty intersection between two signal cells. The cell \(\cell_1\) is an articulation vertex.
\Cref{fig:algo1-right} shows that \(\overlapgraph-\cell_1\) has three connected components:
\(\connectedcomp_1=\{\cell_2,B_1\}\),
\(\connectedcomp_2=\{B_2\}\), and
\(\connectedcomp_3=\{\cell_3,B_3\}\).
By \Cref{alg:det-k2-finest}, the cell \(\cell_1\) is therefore split into the branch pieces
\(\cell_1^{\connectedcomp_1}=\cell_1\cap B_1=\{\fea_1\}\),
\(\cell_1^{\connectedcomp_2}=\cell_1\cap B_2=\{\fea_2,\fea_3\}\), and
\(\cell_1^{\connectedcomp_3}=\cell_1\cap B_3=\{\fea_4\}\).
The remaining signal cells \(\cell_2=\{\fea_5\}\) and \(\cell_3=\{\fea_6\}\) are unchanged. Hence the finest constructible deterministic refinement of \(\givenclassifier_1\) is $\widehat{\partition}_1
=
\bigl\{\{\fea_1\},\{\fea_2,\fea_3\},\{\fea_4\},\{\fea_5\},\{\fea_6\}\bigr\}$.
}
\label{fig:overlap-graph-algorithm-one}
\end{figure}

\begin{lemma}[Branch decomposition of one input cell]
\label{lem:det-k2-branch}
Fix a cell $\cell\in \partition_{\givenclassifier_1}$.
For each connected component $\connectedcomp$ of $\overlapgraph-\cell$, define the $\connectedcomp$-branch of the cell $\cell$ as follows:
\begin{align*}
    \cell_{\connectedcomp}
    \triangleq
    \cell \cap
    \left(
        \bigcup\nolimits_{\cell'\in \vertexset(\connectedcomp)\cap \partition_{\givenclassifier_2}}
        \cell'
    \right)~.
\end{align*}
Then the following statements hold.
\begin{enumerate}
    \item[(i)]
    The nonempty sets $\cell_{\connectedcomp}$ form a partition of $\cell$.

    \item[(ii)]
    For every connected component $\connectedcomp$ of $\overlapgraph-\cell$,
    \begin{align}
    \label{eq:det-k2-branch-indicator}
        \subsetindi_{\cell_{\connectedcomp}}
        =
        \sum\nolimits_{\cell'\in \vertexset(\connectedcomp)\cap \partition_{\givenclassifier_2}}
        \subsetindi_{\cell'}
        -
        \sum\nolimits_{\cell\primed\in \vertexset(\connectedcomp)\cap \partition_{\givenclassifier_1}}
        \subsetindi_{\cell\primed}~.
    \end{align}
    Consequently, $\subsetindi_{\cell_{\connectedcomp}}\in \rowspace(\mymatrix)$, so $\cell_{\connectedcomp}$ is observable by \Cref{cor:identifiable-subset}. Moreover, its observable cell label mass is given by
    \begin{align}
    \label{eq:det-k2-branch-label}
        \feasubsetlabel(\subsetindi_{\cell_{\connectedcomp}})
        =
        \sum\nolimits_{\cell'\in \vertexset(\connectedcomp)\cap \partition_{\givenclassifier_2}}
        \inputcelllabel(\cell')
        -
        \sum\nolimits_{\cell\primed\in \vertexset(\connectedcomp)\cap \partition_{\givenclassifier_1}}
        \inputcelllabel(\cell\primed)~.
    \end{align}

    \item[(iii)]
    If $\feasubset\subseteq \cell$ and $\subsetindi_{\feasubset}\in \rowspace(\mymatrix)$, then there exists a subcollection $\mathcal{T}$ of the connected components of $\overlapgraph-\cell$ such that $\feasubset
    =
    \bigcup\nolimits_{\connectedcomp\in\mathcal{T}} \cell_{\connectedcomp}$.
\end{enumerate}
The completely symmetric statements hold after interchanging the roles of $\partition_{\givenclassifier_1}$ and $\partition_{\givenclassifier_2}$.
\end{lemma}

\begin{algorithm}[ht]
\caption{\textsc{ComputeFinestRefinement}$(j)$ for $j\in\{1,2\}$}
\label{alg:det-k2-finest}
\KwIn{Feature distribution $\feaprob=(\feaprob_i)_{i\in[\feanum]}$ and deterministic predictors $\givenclassifier_1,\givenclassifier_2$.}
\KwOut{A constructible deterministic predictor $\classifier_j^\star\in \classifierfamdet$.}

Construct the signal partitions $\partition_{\givenclassifier_1}=\{\cell_{1,\ell}\}_{\ell\in[\cellnum_1]}$ and $\partition_{\givenclassifier_2}=\{\cell_{2,r}\}_{r\in[\cellnum_2]}$\;

Construct the overlap graph $\overlapgraph$ with vertex set $\partition_{\givenclassifier_1}\sqcup \partition_{\givenclassifier_2}$ and edge $(\cell_{1,\ell},\cell_{2,r})$ iff $\cell_{1,\ell}\cap \cell_{2,r}\neq \emptyset$\;

For every input signal cell $\cell\in \partition_{\givenclassifier_1}\cup \partition_{\givenclassifier_2}$, compute and store
$\inputcelllabel(\cell)\gets \feasubsetmean(\subsetindi_{\cell})\,\feaprob(\subsetindi_{\cell})$\;

Set $j\primed\gets 3-j$\;

Initialize $\refinedpartition_j\gets \emptyset$\;

\ForEach{$\cell\in \partition_{\givenclassifier_j}$}{
    Compute the connected components of $\overlapgraph-\cell$\;
    \ForEach{connected component $\connectedcomp$ of $\overlapgraph-\cell$}{
        $\cellnew \gets
        \cell \cap
        \left(
            \bigcup\nolimits_{\cell\primed\in \vertexset(\connectedcomp)\cap \partition_{\givenclassifier_{j\primed}}}
            \cell\primed
        \right)$\;
        \If{$\cellnew\neq \emptyset$}{
            compute
            $\feasubsetlabel(\subsetindi_{\cellnew})
            \gets
            \sum\nolimits_{\cell\primed\in \vertexset(\connectedcomp)\cap \partition_{\givenclassifier_{j\primed}}}
            \inputcelllabel(\cell\primed)
            -
            \sum\nolimits_{\cell\doubleprimed\in \vertexset(\connectedcomp)\cap \partition_{\givenclassifier_j}}
            \inputcelllabel(\cell\doubleprimed)$\;
            
            assign the prediction
            $\feasubsetmean(\subsetindi_{\cellnew})
            \gets
            \frac{\feasubsetlabel(\subsetindi_{\cellnew})}{\feaprob(\subsetindi_{\cellnew})}$\;
                
            add the cell $\cellnew$ to $\refinedpartition_j$\;
        }
    }
}
\Return the deterministic predictor $\classifier_j^\star$ whose signal partition is $\refinedpartition_j$ and whose prediction on each signal cell $\cellnew\in \refinedpartition_j$ is $\feasubsetmean(\subsetindi_{\cellnew})$\;
\end{algorithm}

\begin{proof}[Proof of \Cref{lem:det-k2-branch}]
We prove the three statements in order.

\noindent\emph{Proof of statement (i).}
Fix a cell $\cell\in \partition_{\givenclassifier_1}$.
For each feature index $i\in \cell$, let $\cell_2(i)$ denote the unique cell of $\partition_{\givenclassifier_2}$ that contains $i$.
Since $\cell_2(i)$ is adjacent to $\cell$ in $\overlapgraph$, the vertex $\cell_2(i)$ belongs to a unique connected component $\connectedcomp$ of $\overlapgraph-\cell$, and then $i\in \cell_{\connectedcomp}$.
Thus the union of the nonempty sets $\cell_{\connectedcomp}$ is all of $\cell$.

If $\connectedcomp_1\neq \connectedcomp_2$ are two connected components of $\overlapgraph-\cell$, then
$\vertexset(\connectedcomp_1)\cap \partition_{\givenclassifier_2}$ and
$\vertexset(\connectedcomp_2)\cap \partition_{\givenclassifier_2}$ are disjoint.
Since the cells of $\partition_{\givenclassifier_2}$ are pairwise disjoint, the sets
$\cell_{\connectedcomp_1}$ and $\cell_{\connectedcomp_2}$ are disjoint.
Hence the nonempty sets $\cell_{\connectedcomp}$ form a partition of $\cell$.

\noindent\emph{Proof of statement (ii).}
Fix a connected component $\connectedcomp$ of $\overlapgraph-\cell$.
We verify Eqn.~\eqref{eq:det-k2-branch-indicator} coordinatewise.
Fix a feature index $i\in[\feanum]$, and let $\cell_1(i)$ and $\cell_2(i)$ be the unique cells of
$\partition_{\givenclassifier_1}$ and $\partition_{\givenclassifier_2}$, respectively, that contain $i$.
There are four cases.

\begin{itemize}
    \item
    If $\cell_1(i)=\cell$ and $\cell_2(i)\in \vertexset(\connectedcomp)$, then the right-hand side of Eqn.~\eqref{eq:det-k2-branch-indicator} equals $1$ at coordinate $i$, and by definition $i\in \cell_{\connectedcomp}$.

    \item
    If $\cell_1(i)=\cell$ and $\cell_2(i)\notin \vertexset(\connectedcomp)$, then the right-hand side of Eqn.~\eqref{eq:det-k2-branch-indicator} equals $0$ at coordinate $i$, and $i\notin \cell_{\connectedcomp}$.

    \item
    If $\cell_1(i)\neq \cell$ and $\cell_1(i)\in \vertexset(\connectedcomp)$, then the edge $(\cell_1(i),\cell_2(i))$ lies in $\overlapgraph-\cell$.
    Since $\connectedcomp$ is a connected component of $\overlapgraph-\cell$, this implies $\cell_2(i)\in \vertexset(\connectedcomp)$ as well.
    Thus the right-hand side of Eqn.~\eqref{eq:det-k2-branch-indicator} equals $1-1=0$ at coordinate $i$.
    Also $i\notin \cell$, hence $i\notin \cell_{\connectedcomp}$.

    \item
    If $\cell_1(i)\neq \cell$ and $\cell_1(i)\notin \vertexset(\connectedcomp)$, then again the edge $(\cell_1(i),\cell_2(i))$ lies in $\overlapgraph-\cell$, so $\cell_2(i)\notin \vertexset(\connectedcomp)$.
    Hence the right-hand side of Eqn.~\eqref{eq:det-k2-branch-indicator} equals $0$ at coordinate $i$, and $i\notin \cell_{\connectedcomp}$.
\end{itemize}

These cases exhaust all possibilities, so Eqn.~\eqref{eq:det-k2-branch-indicator} holds.
Because the right-hand side is a linear combination of rows of $\mymatrix$, we obtain $\subsetindi_{\cell_{\connectedcomp}}\in \rowspace(\mymatrix)$.
Applying the linear functional $\feasubsetlabel(\cdot)$ to Eqn.~\eqref{eq:det-k2-branch-indicator} yields Eqn.~\eqref{eq:det-k2-branch-label}.

\noindent\emph{Proof of statement (iii).}
Assume that $\feasubset\subseteq \cell$ and $\subsetindi_{\feasubset}\in \rowspace(\mymatrix)$.
By \Cref{rem:det-special-case-M}, the rows of $\mymatrix$ are exactly the vectors
$\subsetindi_{\cell'}$ for $\cell'\in \partition_{\givenclassifier_1}\cup \partition_{\givenclassifier_2}$.
Hence there exist coefficients $(\alpha_{\cell'})_{\cell'\in \partition_{\givenclassifier_1}}$ and
$(\beta_{\cell''})_{\cell''\in \partition_{\givenclassifier_2}}$ such that
\begin{align}
\label{eq:det-k2-rowspace-expansion}
    \subsetindi_{\feasubset}
    =
    \sum\nolimits_{\cell'\in \partition_{\givenclassifier_1}} \alpha_{\cell'} \subsetindi_{\cell'}
    +
    \sum\nolimits_{\cell''\in \partition_{\givenclassifier_2}} \beta_{\cell''} \subsetindi_{\cell''}~.
\end{align}

Now let $(\cell\primed,\cell')$ be any edge of $\overlapgraph$ with $\cell\primed\neq \cell$.
Every feature in the nonempty intersection $\cell\primed\cap \cell'$ lies outside $\cell$, whereas $\feasubset\subseteq \cell$.
Therefore the left-hand side of Eqn.~\eqref{eq:det-k2-rowspace-expansion} vanishes on $\cell\primed\cap \cell'$, and so
\begin{align}
\label{eq:det-k2-edge-balance}
    \alpha_{\cell\primed}+\beta_{\cell'}=0
\end{align}
for every such edge $(\cell\primed,\cell')$.
Fix a connected component $\connectedcomp$ of $\overlapgraph-\cell$.

\begin{itemize}
    \item
    Suppose first that $\vertexset(\connectedcomp)\cap \partition_{\givenclassifier_1}\neq \emptyset$.
    Choose any cell $\cell_0^\dagger\in \vertexset(\connectedcomp)\cap \partition_{\givenclassifier_1}$ and set
    $c_{\connectedcomp}\triangleq \alpha_{\cell_0^\dagger}$.
    By Eqn.~\eqref{eq:det-k2-edge-balance}, every neighbor $\cell'$ of $\cell_0^\dagger$ in $\connectedcomp$ satisfies
    $\beta_{\cell'}=-c_{\connectedcomp}$.
    Applying Eqn.~\eqref{eq:det-k2-edge-balance} repeatedly along paths in the connected component shows that every vertex of
    $\vertexset(\connectedcomp)\cap \partition_{\givenclassifier_1}$ has coefficient $c_{\connectedcomp}$ and every vertex of
    $\vertexset(\connectedcomp)\cap \partition_{\givenclassifier_2}$ has coefficient $-c_{\connectedcomp}$.
    Therefore, for every $\cell'\in \vertexset(\connectedcomp)\cap \partition_{\givenclassifier_2}$, the value of
    $\subsetindi_{\feasubset}$ on the whole intersection $\cell\cap \cell'$ is
    \begin{align*}
        \alpha_{\cell}+\beta_{\cell'}
        =
        \alpha_{\cell}-c_{\connectedcomp},
    \end{align*}
    which depends only on $\connectedcomp$ and not on the particular choice of $\cell'$.
    Hence $\subsetindi_{\feasubset}$ is constant on $\cell_{\connectedcomp}$.

    \item
    Suppose next that $\vertexset(\connectedcomp)\cap \partition_{\givenclassifier_1}=\emptyset$.
    Since $\overlapgraph-\cell$ is bipartite, every edge of $\overlapgraph-\cell$ joins a vertex in
    $\partition_{\givenclassifier_1}\setminus \{\cell\}$ to a vertex in $\partition_{\givenclassifier_2}$.
    Because $\connectedcomp$ contains no vertex of $\partition_{\givenclassifier_1}$, it contains no edge.
    As $\connectedcomp$ is connected, it must therefore consist of a single isolated vertex
    $\cell_0\in \partition_{\givenclassifier_2}$.
    In this case, $\cell_{\connectedcomp}=\cell\cap \cell_0$.
    Fix any feature index $i\in \cell_{\connectedcomp}$.
    Since $i\in \cell$ and $i\in \cell_0$, while $i$ belongs to no other cell of either partition, Eqn.~\eqref{eq:det-k2-rowspace-expansion} gives $(\subsetindi_{\feasubset})_i
    =
    \alpha_{\cell}+\beta_{\cell_0}$.
    Thus, we have $\subsetindi_{\feasubset}$ is constant on $\cell_{\connectedcomp}$.
\end{itemize}
Thus, in either case, $\subsetindi_{\feasubset}$ is constant on each branch piece $\cell_{\connectedcomp}$.
Since $\subsetindi_{\feasubset}$ is a $\{0,1\}$-vector, each $\cell_{\connectedcomp}$ is either entirely contained in $\feasubset$ or entirely disjoint from $\feasubset$.
Therefore $\feasubset$ is a union of some of the sets $\cell_{\connectedcomp}$.
\end{proof}

\begin{algorithm}[ht]
\caption{Polynomial-time algorithm for the deterministic version of {\predicrefinelink} when $\classifiernum=2$}
\label{alg:det-k2-solve}
\KwIn{Feature distribution $\feaprob=(\feaprob_i)_{i\in[\feanum]}$ and deterministic predictors $\givenclassifier_1,\givenclassifier_2$.}
\KwOut{Deterministic predictors $(\classifier_1,\classifier_2)$.}

$\classifier_1^\star \gets \textsc{ComputeFinestRefinement}(1)$\;
$\classifier_2^\star \gets \textsc{ComputeFinestRefinement}(2)$\;

\For{$j\in\{1,2\}$}{
    \eIf{$\partition_{\classifier_j^\star}=\partition_{\givenclassifier_j}$}{
        $\classifier_j\gets \givenclassifier_j$\;
    }{
        $\classifier_j\gets \classifier_j^\star$\;
    }
}
\Return $(\classifier_1,\classifier_2)$\;
\end{algorithm}

\begin{proof}[Proof of \Cref{thm:det-k2-polytime}]
We prove the statement for $j=1$; the proof for $j=2$ is completely symmetric.

For each cell $\cell\in \partition_{\givenclassifier_1}$ and each connected component $\connectedcomp$ of $\overlapgraph-\cell$, define $\cell_{\connectedcomp}$ as in \Cref{lem:det-k2-branch}, and set
\begin{align*}
    \refinedpartition_1
    \triangleq
    \left\{
        \cell_{\connectedcomp}\neq \emptyset:
        \ \cell\in \partition_{\givenclassifier_1},
        \ \connectedcomp \text{ is a connected component of } \overlapgraph-\cell
    \right\}.
\end{align*}

For each fixed $\cell\in \partition_{\givenclassifier_1}$, \Cref{lem:det-k2-branch}(i) shows that the nonempty sets $\cell_{\connectedcomp}$ partition $\cell$.
Since the cells of $\partition_{\givenclassifier_1}$ partition $[\feanum]$, it follows that $\refinedpartition_1$ is a partition of $[\feanum]$.
By \Cref{lem:det-k2-branch}(ii), every cell $\cell'\in \refinedpartition_1$ satisfies
$\subsetindi_{\cell'}\in \rowspace(\mymatrix)$, and its quantity
$\feasubsetlabel(\subsetindi_{\cell'})$ is computable from the observed data.
Therefore, by \Cref{cor:constructible-det}, the predictor $\classifier_1^\star$ returned by
Algorithm~\ref{alg:det-k2-finest} belongs to $\classifierfamdet$ and has signal partition
$\partition_{\classifier_1^\star}=\refinedpartition_1$.

We next show that $\refinedpartition_1$ is the finest constructible deterministic refinement of $\partition_{\givenclassifier_1}$.
Let $\classifier\in \classifierfamdet$ be any deterministic predictor such that
$\partition_{\classifier}$ refines $\partition_{\givenclassifier_1}$.
Fix any cell $\cell'\in \partition_{\classifier}$.
Since $\partition_{\classifier}$ refines $\partition_{\givenclassifier_1}$, there exists a unique cell
$\cell\in \partition_{\givenclassifier_1}$ such that $\cell'\subseteq \cell$.
Because $\classifier\in \classifierfamdet$, \Cref{cor:constructible-det} gives
$\subsetindi_{\cell'}\in \rowspace(\mymatrix)$.
Applying \Cref{lem:det-k2-branch}(iii) with $\feasubset=\cell'$ shows that $\cell'$ is a union of cells of $\refinedpartition_1$.
Since this holds for every $\cell'\in \partition_{\classifier}$, we conclude that
\begin{align}
\label{eq:det-k2-finest}
    \text{if }\partition_{\classifier}\text{ refines }\partition_{\givenclassifier_1},
    \text{ then }\refinedpartition_1\text{ refines }\partition_{\classifier}.
\end{align}

We now verify the correctness of Algorithm~\ref{alg:det-k2-solve} for the first input predictor.
\begin{itemize}
    \item 
    \noindent\emph{Case 1: \(\refinedpartition_1\neq \partition_{\givenclassifier_1}\).}
    Then $\partition_{\classifier_1^\star}=\refinedpartition_1$ strictly refines
    $\partition_{\givenclassifier_1}$.
    By \Cref{lem:det-partition-dom}, this implies
    $\classifier_1^\star \partitiondom \givenclassifier_1$.
    
    We claim that $\classifier_1^\star\in \partitiondomparedet$.
    Suppose not.
    Then there exists some predictor $\classifier\in \classifierfamdet$ such that
    $\classifier \partitiondom \classifier_1^\star$.
    By \Cref{lem:det-partition-dom}, the partition $\partition_{\classifier}$ strictly refines
    $\partition_{\classifier_1^\star}=\refinedpartition_1$.
    In particular, $\partition_{\classifier}$ refines $\partition_{\givenclassifier_1}$.
    Applying Eqn.~\eqref{eq:det-k2-finest} to $\classifier$ yields that
    $\refinedpartition_1$ refines $\partition_{\classifier}$, contradicting the strict refinement of
    $\refinedpartition_1$ by $\partition_{\classifier}$.
    Hence $\classifier_1^\star\in \partitiondomparedet$.
    
    Therefore, in this case, Algorithm~\ref{alg:det-k2-solve} outputs
    $\classifier_1=\classifier_1^\star$, and this predictor satisfies
    $\classifier_1\in \partitiondomparedet$ and
    $\classifier_1 \partitiondom \givenclassifier_1$.

    \item 
    \noindent\emph{Case 2: \(\refinedpartition_1= \partition_{\givenclassifier_1}\).}
    We claim that no predictor in $\classifierfamdet$ strictly informationally-dominates $\givenclassifier_1$.
    Suppose, toward a contradiction, that there exists
    $\classifier\in \classifierfamdet$ such that
    $\classifier \partitiondom \givenclassifier_1$.
    By \Cref{lem:det-partition-dom}, the partition $\partition_{\classifier}$ strictly refines
    $\partition_{\givenclassifier_1}$.
    Applying Eqn.~\eqref{eq:det-k2-finest} to $\classifier$ gives that
    $\refinedpartition_1$ refines $\partition_{\classifier}$.
    Since $\refinedpartition_1=\partition_{\givenclassifier_1}$, this implies that
    $\partition_{\givenclassifier_1}$ refines $\partition_{\classifier}$, contradicting the strict refinement of
    $\partition_{\givenclassifier_1}$ by $\partition_{\classifier}$.
    Thus there exists no $\classifier\in \classifierfamdet$ with
    $\classifier \partitiondom \givenclassifier_1$; in particular, there exists no
    $\classifier\in \partitiondomparedet$ with
    $\classifier \partitiondom \givenclassifier_1$.
    Hence Algorithm~\ref{alg:det-k2-solve} correctly outputs
    $\classifier_1=\givenclassifier_1$ in this case.
\end{itemize}
Combining the two cases, the output for the first input predictor satisfies exactly the required specification of {\predicrefinelink}.
By symmetry, the same conclusion holds for the second input predictor.
Therefore Algorithm~\ref{alg:det-k2-solve} solves the deterministic version of {\predicrefinelink} when $\classifiernum=2$.

It remains to prove the running-time bound.
Let $\vertexsetsize
\triangleq
|\vertexset(\overlapgraph)|
=
\cellnum_1+\cellnum_2$
and let $\edgesize$ denote the number of edges of $\overlapgraph$.
Because every feature index belongs to exactly one cell of $\partition_{\givenclassifier_1}$ and exactly one cell of $\partition_{\givenclassifier_2}$, each feature contributes to at most one edge of $\overlapgraph$.
Hence $\edgesize\le \feanum$.
Also $\cellnum_1,\cellnum_2\le \feanum$, so $\vertexsetsize\le 2\feanum$.

The overlap graph and the quantities $\inputcelllabel(\cell)$ for the input cells can be computed in polynomial time from the input.
In one call to Algorithm~\ref{alg:det-k2-finest}, we compute connected components of $\overlapgraph-\cell$ once for each of the at most $\vertexsetsize$ input cells $\cell$; each such computation takes time $O(\vertexsetsize+\edgesize)$.
After the component labels are known, the branch pieces for this fixed $\cell$ can be formed by scanning the feature indices in $\cell$ and grouping them according to the connected component containing their opposite-side cell; this takes time $O(\feanum)$.
Thus one call to Algorithm~\ref{alg:det-k2-finest} runs in time
$O\!\left(\vertexsetsize(\vertexsetsize+\edgesize+\feanum)\right)$
which is polynomial in the input size because $\vertexsetsize,\edgesize\le O(\feanum)$.

Finally, all arithmetic used to compute the quantities
$\inputcelllabel(\cell)$,
$\feasubsetlabel(\subsetindi_{\cell'})$, and
$\feasubsetmean(\subsetindi_{\cell'})$
involves only polynomially many additions, subtractions, multiplications, and divisions of rational numbers whose bit-length is polynomial in the input size.
Hence the total bit complexity is polynomial as well.
The second call to Algorithm~\ref{alg:det-k2-finest} has the same complexity.
Therefore Algorithm~\ref{alg:det-k2-solve} runs in polynomial time.
\end{proof}


\section{Efficient Algorithms for Randomized Predictors}
\label{sec:random}

\newcommand{\signalcone}{\mathcal{M}}
\newcommand{\residualvec}{\mathbf{r}}
\newcommand{\ray}{\operatorname{ray}}
\newcommand{\feasubsetcur}{\feasubset_{\mathrm{cur}}}
\newcommand{\LPFeasiblefeasubset}{\text{\hyperref[LP-Feasible]{$\textsc{Feasible}(\feasubset)$}}}
\newcommand{\LPFeasiblefeasubsetcur}{\text{\hyperref[LP-Feasible]{$\textsc{Feasible}(\feasubsetcur)$}}}
\newcommand{\LPFeasiblefeasubsetcurminus}{\text{\hyperref[LP-Feasible]{$\textsc{Feasible}(\feasubsetcur\setminus\{i\})$}}}
\newcommand{\LPFeasiblefeasubsetsupp}{\text{\hyperref[LP-Feasible]{$\textsc{Feasible}(\supp(\signalvec^\star))$}}}
\newcommand{\LPFeasiblefeasubsetnew}{\text{\hyperref[LP-Feasible]{$\textsc{Feasible}(\feasubset\primed)$}}}

\newcommand{\feasibleprog}[1]{\hyperref[LP-Feasible]{\textsc{Feasible}(#1)}}

In this section, we show that once randomized outputs are allowed, the problem {\predicrefinelink} becomes polynomial-time solvable.
For complexity statements, we assume that all entries of $\feaprob$, all entries of the input signal vectors,
and all input signal predictions are rational numbers encoded in binary.

\begin{theorem}
\label{thm:poly randomized partition}
There exists an algorithm (see \Cref{alg:randomized-predictor-refine}) that runs in time polynomial in the input size and solves {\predicrefinelink}.
\end{theorem}

At a high level, \Cref{alg:randomized-predictor-refine} works by refining each input reduced bag
$\signalbag(\givenclassifier_j)=\{\signalvec_{j,\ell}\}_{\ell\in[\cellnum_j]}$
into a finer bag of constructible signal vectors that still sums to the same row.
The key geometric observation is that all constructible signal vectors lie in the same cone $\signalcone$, which we refer to as the constructible signal cone, and that the atomic directions of this cone are its extreme rays:
\begin{align*}
    \signalcone
    \triangleq
    \rowspace(\mymatrix)\cap \reals_+^{\feanum}~.
\end{align*}
Starting from an input row $\signalvec_{j,\ell}\in \signalcone\setminus\{\vzero\}$, the algorithm repeatedly
decomposes $\signalvec_{j,\ell}$ into finer vectors in $\signalcone$ until every remaining piece lies on an
extreme ray of $\signalcone$.
Such a refinement can only increase information, because one can always
forget the finer signal and map each refined piece back to its parent signal through a feature-independent
garbling.
The geometric part of the argument is carried out in \Cref{lem:support-minimal-extreme-ray,lem:find-extreme-signal-correct,lem:extreme-ray-decompose-correct}, which show that extreme-ray
pieces can be found via support-minimality and obtained in polynomial time through linear-feasibility tests.
The final step is \Cref{lem:extreme-bag-undominated}: once every row of a reduced bag lies on an extreme ray of $\signalcone$,
the resulting constructible predictor already belongs to $\partitiondompare$.
Consequently, for each input predictor $\givenclassifier_j$, the algorithm either leaves
$\signalbag(\givenclassifier_j)$ unchanged, which certifies that $\givenclassifier_j\in\partitiondompare$, or
produces a strict extreme-ray refinement, which yields a constructible predictor
$\classifier_j\in\partitiondompare$ with $\classifier_j\partitiondom \givenclassifier_j$.
We now formalize this geometric viewpoint.

\begin{definition}[Extreme ray]
Let $\signalvec\in \signalcone\setminus\{\vzero\}$.
We write $\ray(\signalvec)\triangleq \{\scalescalar\,\signalvec : \scalescalar\ge 0\}$.
The ray $\ray(\signalvec)$ is called an extreme ray of $\signalcone$ if whenever
$\signalvec = \signalvec\primed+\signalvec\doubleprimed$ with $\signalvec\primed, \signalvec\doubleprimed\in \signalcone$, then
$\signalvec\primed, \signalvec\doubleprimed\in \ray(\signalvec)$.
\end{definition}

\begin{lemma}[Support-minimal vectors span exactly the extreme rays of $\signalcone$]
\label{lem:support-minimal-extreme-ray}
Let $\signalvec\in \signalcone\setminus\{\vzero\}$.
Then the following are equivalent:
\begin{enumerate}[label=(\roman*)]
    \item the ray $\ray(\signalvec)$ is extreme in $\signalcone$;
    \item there exists no $\vu\in \signalcone\setminus\{\vzero\}$ such that
    $\supp(\vu)\subsetneq \supp(\signalvec)$.
\end{enumerate}
\end{lemma}

\Cref{lem:support-minimal-extreme-ray} reduces the problem of finding extreme rays of the constructible signal cone $\signalcone$
to a support-minimality problem. Indeed, by the lemma, a nonzero vector $\vv\in \signalcone$ spans an extreme ray if and only if there is no other nonzero vector in $\signalcone$ whose support is a strict subset of $\supp(\vv)$. This suggests a natural algorithmic approach: rather than searching for extreme rays directly, we may start from a candidate support set $\feasubset\subseteq [\feanum]$ and ask whether there exists any nonzero constructible signal vector supported inside $\feasubset$. If such a feasibility question can be answered efficiently, then we can iteratively shrink supports until no further reduction is possible, at which point \Cref{lem:support-minimal-extreme-ray} guarantees that the resulting vector lies on an extreme ray. This is precisely the role of the following linear-feasibility program \ref{LP-Feasible}:
For every nonempty set $\feasubset \subseteq[\feanum]$, consider the LP-feasibility problem
\begin{equation}
\label{LP-Feasible}
\tag{$\textsc{Feasible}(\feasubset)$}
\begin{aligned}
\text{find } \quad &(\xi,\signalvec)\in \reals^{\totalcells}\times \reals_+^{\feanum} && \\
\text{s.t. } \quad 
&\signalvec=\mymatrix^\top \xi~, &&\\
&\signalvec_i=0~, && i\notin \feasubset~, \\
&\sum\nolimits_{i\in \feasubset}\signalvec_i=1~. &&
\end{aligned}
\end{equation}
This system is feasible if and only if there exists a vector in $\signalcone\setminus\{\vzero\}$ whose support is contained in $\feasubset$.

\begin{algorithm}[H]
\caption{\textsc{FindExtremeSignal}$(\feasubset)$}
\label{alg:find-extreme-signal}
\KwIn{A nonempty set $\feasubset\subseteq[\feanum]$ such that $\feasibleprog{\feasubset}$ is feasible.}
\KwOut{A vector $\signalvec^\star\in \signalcone\setminus\{\vzero\}$ whose ray is extreme.}

$\feasubsetcur\gets \feasubset$\;

\ForEach{$i\in \feasubset$}{
    \If{$i\in \feasubsetcur$ and $\feasibleprog{\feasubsetcur\setminus\{i\}}$ is feasible}{
        $\feasubsetcur\gets \feasubsetcur\setminus\{i\}$\;
    }
}

Solve $\feasibleprog{\feasubsetcur}$ and return any feasible vector $\signalvec^\star$\;
\end{algorithm}

\begin{lemma}
\label{lem:find-extreme-signal-correct}
\Cref{alg:find-extreme-signal} is well defined.
If program $\LPFeasiblefeasubset$ is feasible, then the returned vector $\signalvec^\star$ belongs to
$\signalcone\setminus\{\vzero\}$ and the ray $\ray(\signalvec^\star)$ is extreme in $\signalcone$.
\end{lemma}

\begin{proof}
Because the loop only deletes one index at a time, it terminates after at most $|\feasubset|$ iterations.
By construction, program $\feasibleprog{\feasubsetcur}$ is feasible at the moment the algorithm stops,
so the returned vector $\signalvec^\star$ belongs to $\signalcone\setminus\{\vzero\}$.

We claim that $\supp(\signalvec^\star)=\feasubsetcur$.
Indeed, $\supp(\signalvec^\star)\subseteq \feasubsetcur$ by the constraints of $\feasibleprog{\feasubsetcur}$.
If the inclusion were strict, then program~$\LPFeasiblefeasubsetsupp$ would also be feasible,
because the same vector $\signalvec^\star$ would satisfy all the constraints after replacing $\feasubsetcur$
by $\supp(\signalvec^\star)$.
Choose any $i\in \feasubsetcur\setminus \supp(\signalvec^\star)$.
Then $\supp(\signalvec^\star)\subseteq \feasubsetcur\setminus\{i\}$, so
$\feasibleprog{\feasubsetcur\setminus\{i\}}$ would be feasible, contradicting the stopping rule.
Thus, $\supp(\signalvec^\star)=\feasubsetcur$.

We next claim that there is no nonzero vector of $\signalcone$ supported on a strict subset of
$\feasubsetcur$.
If such a vector existed, then after normalizing it to have coordinate sum $1$ on its support, program
$\LPFeasiblefeasubsetnew$ would be feasible for some strict subset $\feasubset\primed \subsetneq \feasubsetcur$.
Choose any $i\in \feasubsetcur\setminus \feasubset\primed$.
Then $\feasibleprog{\feasubsetcur\setminus\{i\}}$ would be feasible, again contradicting the stopping rule.
Thus, signal vector $\signalvec^\star$ satisfies condition (ii) of \Cref{lem:support-minimal-extreme-ray},
and hence its ray is extreme by that lemma.
\end{proof}

\begin{algorithm}[H]
\caption{\textsc{ExtremeRayDecompose}$(\signalvec)$}
\label{alg:extreme-ray-decompose}
\KwIn{A vector $\signalvec\in \signalcone\setminus\{\vzero\}$.}
\KwOut{A finite bag $\mathcal{B}(\signalvec)$ of vectors in $\signalcone\setminus\{\vzero\}$ whose rays are extreme and whose sum is $\signalvec$.}
$\residualvec\gets \signalvec$\;
Initialize $\mathcal{B}(\signalvec)\gets \emptyset$, $t\gets 1$\;
\While{$\residualvec\neq \vzero$}{
    $\feasubset\gets \supp(\residualvec)$\;
    $\vu\gets \textsc{FindExtremeSignal}(\feasubset)$\;
    $\scalescalar\gets \min_{i\in \supp(\vu)} \residualvec_i/\vu_i$\;
    $\vv^{(t)}\gets \scalescalar \vu$\;
    add $\vv^{(t)}$ to $\mathcal{B}(\signalvec)$\;
    $\residualvec\gets \residualvec - \vv^{(t)}$\;
    $t\gets t + 1$
}
\Return $\mathcal{B}(\signalvec)$\;
\end{algorithm}





\begin{lemma}
\label{lem:extreme-ray-decompose-correct}
For every $\signalvec\in \signalcone\setminus\{\vzero\}$,
\Cref{alg:extreme-ray-decompose} terminates after at most $\feanum$ iterations and returns a finite bag
$\mathcal{B}(\signalvec)$ consisting of vectors in $\signalcone\setminus\{\vzero\}$, each spanning an extreme ray of $\signalcone$, such that
\begin{align*}
    \sum\nolimits_{\vv\primed\in \mathcal{B}(\signalvec)} \vv\primed
    =
    \signalvec .
\end{align*}
\end{lemma}

\begin{proof}
Fix one iteration and let $\feasubset=\supp(\residualvec)$.
Because $\residualvec\in \signalcone\setminus\{\vzero\}$, the LP $\LPFeasiblefeasubset$ is feasible after scaling $\residualvec$ so that its coordinates on $\feasubset$ sum to $1$.
Thus \Cref{alg:find-extreme-signal} is well defined and returns
$\vu\in \signalcone\setminus\{\vzero\}$ spanning an extreme ray of $\signalcone$.
Since $\supp(\vu)\subseteq \feasubset=\supp(\residualvec)$, the scalar
\begin{align*}
    \scalescalar
    \triangleq
    \min_{i\in \supp(\vu)} \frac{\residualvec_i}{\vu_i}
\end{align*}
is well defined and strictly positive.
Therefore $\vv\primed = \scalescalar \vu$ belongs to $\signalcone\setminus\{\vzero\}$ and spans the same extreme ray as $\vu$.
Moreover, $\residualvec-\vv\primed \in \rowspace(\mymatrix)$ and $\residualvec-\vv\primed \ge \vzero$, so
$\residualvec-\vv\primed \in \signalcone$.
By the definition of $\scalescalar$, at least one coordinate of $\residualvec-\vv\primed $ is zero on $\supp(\vu)$.
Thus, the support strictly shrinks whenever $\residualvec\neq \vzero$.
Therefore the loop terminates after at most $\feanum$ iterations.

If the output bag is $\mathcal{B}(\signalvec)=\{\vv^{(1)},\dots,\vv^{(q)}\}$ in the order they are generated,
then the residual vectors satisfy
\begin{align*}
    \residualvec^{(1)}
    = \signalvec~;\quad
    \residualvec^{(t+1)}
    = \residualvec^{(t)}-\vv^{(t)}
    \;\; \text{for } t\in[q]~;
    \quad
    \residualvec^{(q+1)}
     = \vzero~.
\end{align*}
Summing these identities yields 
$\signalvec
=
\sum\nolimits_{t\in[q]} \vv^{(t)}$,
as claimed.
\end{proof}

The previous lemmas show that every constructible signal vector in the cone $\signalcone$
admits a finite decomposition into extreme-ray pieces.
We next connect this geometric decomposition to informational dominance.
For this purpose, it is convenient to use the non-strict garbling preorder.

For two finitely supported predictors $\givenclassifier$ and $\classifier$, we write
$\givenclassifier \partitiondomeq \classifier$
if there exists a feature-independent stochastic kernel
$\mapkernal:\signalspace_{\givenclassifier}\to \Delta(\signalspace_{\classifier})$
such that, for every feature $\fea\in\feaspace$, the conditional signal law of $\classifier$
is obtained from that of $\givenclassifier$ by passing through $\mapkernal$.
By \Cref{lem:partition dom bag equiv}, this is equivalent to the existence of the corresponding
row-stochastic matrix in bag form.
Moreover,
\begin{align*}
    \givenclassifier \partitiondom \classifier
    \quad\Longleftrightarrow\quad
    \givenclassifier \partitiondomeq \classifier
    \text{ and }
    \classifier \not\partitiondomeq \givenclassifier~.
\end{align*}
The preorder $\partitiondomeq$ is transitive:
if $\classifier_1 \partitiondomeq \classifier_2$ via a feature-independent garbling
$\mapkernal_{12}$ and $\classifier_2 \partitiondomeq \classifier_3$ via a feature-independent garbling
$\mapkernal_{23}$, then the composed channel $\mapkernal_{23}\circ \mapkernal_{12}$ is again
feature-independent and maps the signal law of $\classifier_1$ directly to that of $\classifier_3$.
We now show that once every row of a reduced bag lies on an extreme ray of $\signalcone$,
the resulting constructible predictor is already informationally-undominated.


\begin{lemma}[Extreme-ray bags are informationally-undominated]
\label{lem:extreme-bag-undominated}
Let $\classifier\in \classifierfam$ be a constructible predictor, and write its reduced bag as
$\signalbag(\classifier)=\{\signalvec_{\classifier,r}\}_{r\in[\cellnum_{\classifier}]}$.
Assume that every vector $\signalvec_{\classifier,r}$ spans an extreme ray of $\signalcone$.
Then $\classifier\in \partitiondompare$.
More precisely, for every predictor $h\in \classifierfam$, if we have $h\partitiondomeq \classifier$, then we have $\signalbag(h)=\signalbag(\classifier)$ as reduced bags, up to reordering.
\end{lemma}

\begin{algorithm}[H]
\caption{Solving {\predicrefinelink} for randomized predictors}
\label{alg:randomized-predictor-refine}
\KwIn{Feature distribution $\feaprob=(\feaprob_i)_{i\in[\feanum]}$ and finitely supported predictors $(\givenclassifier_j)_{j\in[\classifiernum]}$.}
\KwOut{Predictors $(\classifier_j)_{j\in[\classifiernum]}$.}
Compute the reduced bags
$\signalbag(\givenclassifier_j)=\{\signalvec_{j,\ell}\}_{\ell\in[\cellnum_j]}$
for all $j\in[\classifiernum]$\;
Construct the signal-feature matrix $\mymatrix$ and the observable vector $\weightedpredic$
as in \Cref{def:signal-feature-matrix}\;
\For{$j\in[\classifiernum]$}{
    initialize an empty bag $\widetilde{\signalbag}_j$\;
    \ForEach{$\ell\in[\cellnum_j]$}{
        run \textsc{ExtremeRayDecompose}$(\signalvec_{j,\ell})$ and add all returned vectors to $\widetilde{\signalbag}_j$\;
    }
    let $\signalbag_j^\star$ be the ray-reduction of $\widetilde{\signalbag}_j$\;
    \ForEach{$\signalvec\in \signalbag_j^\star$}{
        solve $\mymatrix^\top \xi = \signalvec$\;
        compute
        $\prediction(\signalvec)\gets \dfrac{\xi^\top \weightedpredic}{\feasubsetprob(\signalvec)}$\;
    }
    let $\classifier_j^\star$ be the predictor whose reduced bag is $\signalbag_j^\star$ and whose prediction on each signal $\signalvec\in \signalbag_j^\star$ is $\prediction(\signalvec)$\;
    \uIf{$\signalbag_j^\star=\signalbag(\givenclassifier_j)$}{
        $\classifier_j\gets \givenclassifier_j$\;
    }
    \Else{
        $\classifier_j\gets \classifier_j^\star$\;
    }
}
\Return $(\classifier_j)_{j\in[\classifiernum]}$\;
\end{algorithm}

We are now ready to prove \Cref{thm:poly randomized partition}.
\begin{proof}[Proof of \Cref{thm:poly randomized partition}]
We first verify that every output constructed by the algorithm is well defined.
Fix $j\in[\classifiernum]$ and fix one input row $\signalvec_{j,\ell}\in \signalbag(\givenclassifier_j)$.
Since predictor $\givenclassifier_j$ itself is constructible from the observable information of the given instance,
\Cref{def:constructible-predictor} implies that $\signalvec_{j,\ell}\in \signalcone\setminus\{\vzero\}$.
Thus, \Cref{alg:extreme-ray-decompose} applies and returns a finite bag
$\mathcal{B}(\signalvec_{j,\ell})$ of vectors in $\signalcone\setminus\{\vzero\}$ whose sum is
$\signalvec_{j,\ell}$, by \Cref{lem:extreme-ray-decompose-correct}.
Since $\feaprob_i>0$ for every $i$ and every output vector is nonzero and nonnegative, we have
$\feasubsetprob(\signalvec)>0$ for every $\signalvec\in \widetilde{\signalbag}_j$ and for every
$\signalvec\in \signalbag_j^\star$.
Thus every prediction $\prediction(\signalvec)$ is well defined.

We next prove that $\classifier_j^\star$ is a constructible predictor.
Every vector in $\widetilde{\signalbag}_j$ belongs to $\signalcone\setminus\{\vzero\}$, hence to
$\rowspace(\mymatrix)\cap \reals_+^{\feanum}\setminus\{\vzero\}$.
By construction, we have
\begin{align*}
    \sum\nolimits_{\signalvec\in \widetilde{\signalbag}_j}\signalvec
    =
    \sum\nolimits_{\ell\in[\cellnum_j]}\signalvec_{j,\ell}
    =
    \one ~.
\end{align*}
The ray-reduction $\signalbag_j^\star$ is obtained by summing vectors inside each ray class,
so every vector in $\signalbag_j^\star$ still belongs to $\rowspace(\mymatrix)\cap \reals_+^{\feanum}\setminus\{\vzero\}$,
and the vectors in $\signalbag_j^\star$ still sum coordinatewise to $\one$.
Thus, \Cref{def:constructible-predictor} implies that $\classifier_j^\star\in \classifierfam$.

We now verify the prediction formula used in the algorithm.
Fix $\signalvec\in \signalbag_j^\star$.
Because $\signalvec\in \rowspace(\mymatrix)$, the linear system
$\mymatrix^\top \xi=\signalvec$ has a solution.
If $\xi$ and $\xi'$ are two such solutions, then
$\mymatrix^\top(\xi-\xi')=\vzero$, so
\begin{align*}
    (\xi-\xi')^\top \weightedpredic
    =
    (\xi-\xi')^\top \mymatrix \weightedbayes
    =
    \bigl(\mymatrix^\top(\xi-\xi')\bigr)^\top \weightedbayes
    =
    0 ~.
\end{align*}
Thus $\xi^\top \weightedpredic$ depends only on $\signalvec$ and not on the chosen solution $\xi$.
Moreover,
\begin{align*}
    \xi^\top \weightedpredic
    =
    \xi^\top \mymatrix \weightedbayes
    =
    (\mymatrix^\top \xi)^\top \weightedbayes
    =
    \signalvec^\top \weightedbayes
    =
    \feasubsetlabel(\signalvec) ~.
\end{align*}
Thus, the algorithm indeed assigns the calibrated value
\begin{align*}
    \prediction(\signalvec)
    =
    \frac{\xi^\top \weightedpredic}{\feasubsetprob(\signalvec)}
    =
    \frac{\feasubsetlabel(\signalvec)}{\feasubsetprob(\signalvec)}
    =
    \feasubsetmean(\signalvec) ~.
\end{align*}
We next show that every vector in $\signalbag_j^\star$ spans an extreme ray of $\signalcone$.
Indeed, every vector in $\widetilde{\signalbag}_j$ does so by
\Cref{lem:extreme-ray-decompose-correct}, and ray-reduction only merges positive multiples
lying on the same ray.

Therefore \Cref{lem:extreme-bag-undominated} gives
\begin{align}
    \classifier_j^\star \in \partitiondompare~.
    \label{eq:random-alg-output-undominated}
\end{align}

We now prove that $\classifier_j^\star \partitiondomeq \givenclassifier_j$.
Let $\widetilde{\classifier}_j$ be the predictor represented by the unreduced bag $\widetilde{\signalbag}_j$.
For every piece $\signalvec\in \mathcal{B}(\signalvec_{j,\ell})$, define a deterministic garbling that maps the signal $\signalvec$ of $\widetilde{\classifier}_j$
to the parent signal $\signalvec_{j,\ell}$ of $\givenclassifier_j$.
Because
$\sum\nolimits_{\signalvec\in \mathcal{B}(\signalvec_{j,\ell})}\signalvec=\signalvec_{j,\ell}$,
this is a valid feature-independent garbling from $\widetilde{\classifier}_j$ to $\givenclassifier_j$.
Thus, we have $\widetilde{\classifier}_j \partitiondomeq \givenclassifier_j$.
By the ray-reduction equivalence proved in \Cref{lem:ray-merging-preserves-dom} in \Cref{apx:proofs random},
passing from the unreduced bag $\widetilde{\signalbag}_j$ to its ray-reduction $\signalbag_j^\star$
does not weaken informational-dominance content.
Thus, we have $\classifier_j^\star \partitiondomeq \widetilde{\classifier}_j$.
By transitivity of $\partitiondomeq$, this leads to
\begin{align}
    \classifier_j^\star \partitiondomeq \givenclassifier_j~.
    \label{eq:random-alg-output-garbles-input}
\end{align}
We now analyze the two cases of the algorithm.
\begin{itemize}
    \item 
    \textbf{Case 1: $\signalbag_j^\star\neq \signalbag(\givenclassifier_j)$.}
    By Eqn.~\eqref{eq:random-alg-output-garbles-input}, we already know
    $\classifier_j^\star \partitiondomeq \givenclassifier_j$.
    We claim that
    $\givenclassifier_j \not\partitiondomeq \classifier_j^\star$.
    Otherwise, applying \Cref{lem:extreme-bag-undominated} to $\classifier_j^\star$
    with $h=\givenclassifier_j$ would imply
    $\signalbag(\givenclassifier_j)=\signalbag_j^\star$,
    contrary to the assumption of this case.
    Thus, we have $\classifier_j^\star \partitiondom \givenclassifier_j$.
    Together with Eqn.~\eqref{eq:random-alg-output-undominated}, this shows that
    the algorithm outputs $\classifier_j=\classifier_j^\star\in \partitiondompare$
    and $\classifier_j\partitiondom \givenclassifier_j$ in this case.

    \item 
    \textbf{Case 2: $\signalbag_j^\star=\signalbag(\givenclassifier_j)$.}
    Since every vector in $\signalbag_j^\star$ spans an extreme ray of $\signalcone$,
    the same is true for every vector in $\signalbag(\givenclassifier_j)$.
    Applying \Cref{lem:extreme-bag-undominated} directly to $\givenclassifier_j$ gives
    $\givenclassifier_j\in \partitiondompare$.
    Thus, there does not exist any predictor $h\in \partitiondompare$ such that
    $h\partitiondom \givenclassifier_j$.
    In this case the algorithm outputs $\classifier_j=\givenclassifier_j$, exactly as required by
    the specification of {\predicrefinelink}.
\end{itemize}
Since the above argument holds for every $j\in[\classifiernum]$, the algorithm is correct.

It remains to prove the running-time bound.
Let $L_{\mathrm{in}}$ denote the total binary encoding length of the input:
namely, of $\feaprob$, of all rational entries of the input kernels,
and of all rational input signal predictions.
From the kernels, the reduced bags $\signalbag(\givenclassifier_j)$ can be computed in polynomial time by first forming the raw vectors and then merging equal normalized ray representatives.
For example, for every nonzero rational vector $\signalvec$, define its normalized representative
\begin{align*}
    \widehat{\signalvec}
    \triangleq
    \frac{\signalvec}{\sum\nolimits_{i\in[\feanum]}\signalvec_i}~.
\end{align*}
Since $\signalvec\in \reals_+^{\feanum}\setminus\{\vzero\}$, two such vectors lie on the same positive ray if and only if their normalized representatives are equal.
Thus both ray-reduction and equality testing of reduced bags are polynomial-time operations over rational data.

The matrix $\mymatrix$ and the vector $\weightedpredic$ of \Cref{def:signal-feature-matrix}
are computable in polynomial time from the reduced bags and the input signal predictions.
Every LP instance $\LPFeasiblefeasubset$ has $\totalcells+\feanum$ variables and a number of linear constraints polynomial in $\totalcells+\feanum$, with rational coefficients of bit-length polynomial in $L_{\mathrm{in}}$.
Therefore every feasibility test and every feasible solution can be computed in time polynomial in $L_{\mathrm{in}}$ by standard polynomial-time algorithms for rational linear programming.

By \Cref{alg:find-extreme-signal}, one call to \textsc{FindExtremeSignal} performs at most $\feanum$ LP-feasibility tests and one final LP solve.
By \Cref{lem:extreme-ray-decompose-correct}, one call to \textsc{ExtremeRayDecompose} performs at most $\feanum$ iterations.
Thus, one call to \textsc{ExtremeRayDecompose} uses at most $O(\feanum^2)$ LP solves.
Since the total number of input reduced-bag rows is
$\totalcells=\sum\nolimits_{j\in[\classifiernum]}\cellnum_j$,
the whole algorithm performs at most $O(\totalcells\,\feanum^2)$ LP solves.

Finally, each input row yields at most $\feanum$ extreme-ray pieces, so the unreduced bag
$\widetilde{\signalbag}_j$ has size at most $\cellnum_j\feanum$, and the final reduced bag
$\signalbag_j^\star$ has size at most $\cellnum_j\feanum$ as well.
For each output row $\signalvec\in \signalbag_j^\star$, the algorithm solves the rational linear system
$\mymatrix^\top \xi=\signalvec$ and computes the rational number
$\prediction(\signalvec)=\xi^\top \weightedpredic/\feasubsetprob(\signalvec)$;
this takes polynomial time.
Therefore the total running time of \Cref{alg:randomized-predictor-refine} is polynomial in the input size $L_{\mathrm{in}}$.

Thus \Cref{alg:randomized-predictor-refine} is a polynomial-time algorithm that solves
{\predicrefinelink} under informational dominance.
\end{proof}

\section*{Acknowledgment}
We thank the anonymous reviewers at NeurIPS 2026 for their helpful and valuable comments and suggestions.



\newpage
\bibliography{mybib}

\newpage
\appendix

\section{Missing Proofs in \Cref{sec:det}}
\label{apx:proofs det}

\begin{proof}[Proof of \Cref{lem:det-partition-dom}]
Suppose first that $\classifier \partitiondom \classifiernew$.
By \Cref{defn:partition dom}, there exists a feature-independent garbling from predictor $\classifier$ to predictor $\classifiernew$.

Fix any signal cell $\cell\in \partition_{\classifier}$.
If two features $\fea_i,\fea_{i'}\in \cell$, then they induce the same signal of $\classifier$.
After passing through the same feature-independent garbling kernel, they therefore induce the same distribution over signals of $\classifiernew$.
Since $\classifiernew$ is deterministic, each of $\fea_i$ and $\fea_{i'}$ induces a point mass on a single signal of $\classifiernew$.
Thus, these point masses must coincide, so $\fea_i$ and $\fea_{i'}$ lie in the same signal cell of $\classifiernew$.
Therefore every cell of $\partition_{\classifier}$ is contained in a cell of $\partition_{\classifiernew}$, and thus $\partition_{\classifier}$ refines $\partition_{\classifiernew}$.

We next show that the refinement is strict.
If the signal partition $\partition_{\classifiernew}$ also refined $\partition_{\classifier}$, then applying the same argument with the roles of $\classifier$ and $\classifiernew$ reversed would imply that there exists a feature-independent garbling from $\classifiernew$ to $\classifier$, contradicting the strictness in \Cref{defn:partition dom}.
Thus, $\partition_{\classifiernew}$ does not refine $\partition_{\classifier}$, and therefore $\partition_{\classifier}$ strictly refines $\partition_{\classifiernew}$.

Conversely, suppose that $\partition_{\classifier}$ strictly refines $\partition_{\classifiernew}$.
Then every cell of $\partition_{\classifier}$ is contained in a unique cell of $\partition_{\classifiernew}$.
Define a feature-independent garbling from $\classifier$ to $\classifiernew$ by mapping each signal of $\classifier$ deterministically to the signal of $\classifiernew$ corresponding to that containing cell.
This yields a feature-independent garbling from $\classifier$ to $\classifiernew$.

It remains to show that there is no feature-independent garbling from $\classifiernew$ to $\classifier$.
If such a garbling existed, then by the first part of the proof, every cell of $\partition_{\classifiernew}$ would be contained in a cell of $\partition_{\classifier}$; that is, $\partition_{\classifiernew}$ would refine $\partition_{\classifier}$.
This contradicts the assumption that $\partition_{\classifier}$ strictly refines $\partition_{\classifiernew}$.
Thus, we have $\classifier \partitiondom \classifiernew$.
\end{proof}

\begin{proof}[Proof of \Cref{cor:hardness predicrefine}]
Assume that there exists a polynomial-time algorithm for {\predicrefinelink}, restricted to deterministic inputs and deterministic outputs.
Given an instance of {\partdecisionlink}, run this algorithm on the same input and inspect the first output predictor $\classifier_1$.
If $\classifier_1=\givenclassifier_1$, output \textsc{No}; otherwise output \textsc{Yes}.

We claim that this procedure is correct.
If the answer to {\partdecisionlink} is \textsc{Yes}, then by definition there exists some predictor $h\in\partitiondomparedet$ such that $h \partitiondom \givenclassifier_1$.
By the specification of {\predicrefinelink}, the algorithm must therefore return a predictor $\classifier_1\in\partitiondomparedet$ such that $\classifier_1 \partitiondom \givenclassifier_1$.
In particular, $\classifier_1\neq \givenclassifier_1$.

Conversely, if the answer to {\partdecisionlink} is \textsc{No}, then there exists no predictor $h\in\partitiondomparedet$ such that $h \partitiondom \givenclassifier_1$.
Again by the specification of {\predicrefinelink}, the algorithm must return $\classifier_1=\givenclassifier_1$.

Thus a polynomial-time algorithm for the deterministic version of {\predicrefinelink} would yield a polynomial-time algorithm for {\partdecisionlink}.
Since {\partdecisionlink} is $\mathsf{NP}$-hard by \Cref{thm:hardness partdecision}, the deterministic version of {\predicrefinelink} is $\mathsf{NP}$-hard under polynomial-time Turing reductions.
\end{proof}

\section{Missing Proofs in \Cref{sec:random}}
\label{apx:proofs random}

\begin{lemma}[Ray merging preserves dominance]
\label{lem:ray-merging-preserves-dom}
Let $\signalbag=\{\signalvec_1,\dots,\signalvec_{\cellnum}\}\subseteq \reals_+^{\feanum}\setminus\{\vzero\}$ be a finite bag
satisfying $\sum\nolimits_{\ell\in[\cellnum]}\signalvec_\ell=\one$.
Let $\signalbagred$ be obtained from $\signalbag$ by merging every $\rayeq$-equivalence class into the sum of all vectors in that class.
Let $\classifier^{\mathrm{raw}}$ and $\classifier^{\mathrm{red}}$ be the predictors represented by
$\signalbag$ and $\signalbagred$, respectively.
Then $\classifier^{\mathrm{raw}} \partitiondomeq \classifier^{\mathrm{red}} $ and $
\classifier^{\mathrm{red}} \partitiondomeq \classifier^{\mathrm{raw}}$.
Consequently, for every finitely supported predictor $\classifiernew$,
\begin{align*}
    \classifiernew\partitiondom \classifier^{\mathrm{raw}}
    \Longleftrightarrow
    \classifiernew\partitiondom \classifier^{\mathrm{red}}~;\quad
    \classifier^{\mathrm{raw}}\partitiondom \classifiernew
    \Longleftrightarrow
    \classifier^{\mathrm{red}}\partitiondom \classifiernew~.
\end{align*}
\end{lemma}

\begin{proof}
Let $\mathcal{R}_1,\dots,\mathcal{R}_L$ be the distinct $\rayeq$-equivalence classes in $\signalbag$, and for each $t\in[L]$ write
\begin{align*}
    \bar{\signalvec}_t
    \triangleq
    \sum\nolimits_{\signalvec_\ell\in \mathcal{R}_t}\signalvec_\ell .
\end{align*}
Thus $\signalbagred=\{\bar{\signalvec}_t\}_{t\in[L]}$.
To obtain $\classifier^{\mathrm{red}}$ from $\classifier^{\mathrm{raw}}$, map every raw signal
$\ell\in[\cellnum]$ deterministically to the unique class $t\in[L]$ with $\signalvec_\ell\in \mathcal{R}_t$.
This is a feature-independent garbling, so
$\classifier^{\mathrm{raw}} \partitiondomeq \classifier^{\mathrm{red}}$.

For the reverse direction, fix a class $\mathcal{R}_t$.
Choose any representative $\signalvec_t^\star\in \mathcal{R}_t$.
For every $\signalvec_\ell\in \mathcal{R}_t$, there exists a unique scalar $\alpha_\ell>0$ such that
$\signalvec_\ell=\alpha_\ell \signalvec_t^\star$.
Thus,
\begin{align*}
    \bar{\signalvec}_t
    =
    \sum\nolimits_{\signalvec_\ell\in \mathcal{R}_t}\alpha_\ell \signalvec_t^\star
    =
    \Bigl(\sum\nolimits_{\signalvec_\ell\in \mathcal{R}_t}\alpha_\ell\Bigr)\signalvec_t^\star .
\end{align*}
Define $\gamma_\ell
\triangleq
\frac{\alpha_\ell}{\sum\nolimits_{\signalvec_r\in \mathcal{R}_t}\alpha_r}$ for each $\signalvec_\ell\in \mathcal{R}_t$.
Then $\gamma_\ell\ge 0$, $\sum\nolimits_{\signalvec_\ell\in \mathcal{R}_t}\gamma_\ell=1$, and
$\signalvec_\ell=\gamma_\ell \bar{\signalvec}_t$ for every $\signalvec_\ell\in \mathcal{R}_t$.
Therefore, conditioned on the merged signal $t$, one can recover a raw signal inside $\mathcal{R}_t$
by an independent random split with probabilities $(\gamma_\ell)_{\signalvec_\ell\in \mathcal{R}_t}$.
Doing this independently for every class yields a feature-independent garbling from
$\classifier^{\mathrm{red}}$ to $\classifier^{\mathrm{raw}}$.
Thus $\classifier^{\mathrm{red}} \partitiondomeq \classifier^{\mathrm{raw}}$.

Finally, because $\partitiondomeq$ is transitive, mutual $\partitiondomeq$-equivalence preserves every strict comparison.
For example, if $\classifiernew\partitiondom \classifier^{\mathrm{raw}}$, then
$\classifiernew\partitiondomeq \classifier^{\mathrm{raw}} \partitiondomeq \classifier^{\mathrm{red}}$, so
$\classifiernew\partitiondomeq \classifier^{\mathrm{red}}$.
If also $\classifier^{\mathrm{red}}\partitiondomeq \classifiernew$, then
$\classifier^{\mathrm{raw}}\partitiondomeq \classifier^{\mathrm{red}}\partitiondomeq \classifiernew$,
contradicting $\classifiernew\partitiondom \classifier^{\mathrm{raw}}$.
Thus $\classifiernew\partitiondom \classifier^{\mathrm{red}}$.
The other implications are identical.
\end{proof}

\begin{proof}[Proof of \Cref{lem:support-minimal-extreme-ray}]
We first prove $(i)\Rightarrow(ii)$.
Assume that $\ray(\signalvec)$ is extreme, and suppose toward a contradiction that there exists
$\vu\in \signalcone\setminus\{\vzero\}$ such that
$\supp(\vu)\subsetneq \supp(\signalvec)$.
Define
\begin{align*}
    \scalescalar
    \triangleq
    \min_{i\in \supp(\vu)} \frac{\signalvec_i}{\vu_i}~.
\end{align*}
This scalar is well-defined and strictly positive because
$\supp(\vu)\subseteq \supp(\signalvec)$ and $\vu_i>0$ for every
$i\in\supp(\vu)$.

By the definition of $\scalescalar$, we have
$\signalvec-\scalescalar\vu\in \reals_+^{\feanum}$.
Moreover, since both $\signalvec$ and $\vu$ belong to $\rowspace(\mymatrix)$, we also have
\begin{align*}
    \signalvec-\scalescalar\vu\in \rowspace(\mymatrix)~.
\end{align*}
Thus, we have $\signalvec-\scalescalar\vu\in\signalcone$.
Furthermore, $\signalvec-\scalescalar\vu\neq \vzero$, because every coordinate in
$\supp(\signalvec)\setminus\supp(\vu)$ remains strictly positive. Thus,
\begin{align*}
    \signalvec
    =
    \scalescalar\vu + (\signalvec-\scalescalar\vu)
\end{align*}
is a decomposition of $\signalvec$ into two nonzero vectors in $\signalcone$.
However, the first summand $\scalescalar\vu$ does not belong to $\ray(\signalvec)$:
it vanishes on $\supp(\signalvec)\setminus\supp(\vu)$, whereas $\signalvec$ is strictly
positive on that set. This contradicts the extremality of $\ray(\signalvec)$.

We next prove $(ii)\Rightarrow(i)$.
Assume (ii), and consider any decomposition
\begin{align*}
    \signalvec
    =
    \signalvec\primed+\signalvec\doubleprimed
    \quad\text{with }
    \signalvec\primed,\signalvec\doubleprimed\in\signalcone~.
\end{align*}
If either $\signalvec\primed=\vzero$ or $\signalvec\doubleprimed=\vzero$, then both summands
belong to $\ray(\signalvec)$, so there is nothing to prove. Hence assume
$\signalvec\primed,\signalvec\doubleprimed\neq\vzero$.

Since $\signalvec\primed,\signalvec\doubleprimed\ge\vzero$ and
$\signalvec=\signalvec\primed+\signalvec\doubleprimed$, we have
$\supp(\signalvec\primed)\subseteq \supp(\signalvec),
\supp(\signalvec\doubleprimed)\subseteq \supp(\signalvec)$.
By condition (ii), neither inclusion can be strict. Thus, we have
$\supp(\signalvec\primed)=\supp(\signalvec)$ and $
\supp(\signalvec\doubleprimed)=\supp(\signalvec)$.

We show that $\signalvec\primed\in\ray(\signalvec)$. Suppose, toward a contradiction, that
$\signalvec\primed\notin\ray(\signalvec)$. Since both summands have the same support as
$\signalvec$, for every $i\in\supp(\signalvec)$ we have $0<\signalvec\primed_i<\signalvec_i$.
Define
\begin{align*}
    \scalescalar\primed
    \triangleq
    \min_{i\in \supp(\signalvec)}
    \frac{\signalvec\primed_i}{\signalvec_i}
    \in (0,1)~,
    \quad
    \vz
    \triangleq
    \signalvec\primed-\scalescalar\primed\signalvec~.
\end{align*}
By construction, $\vz\ge\vzero$. Since
$\signalvec\primed,\signalvec\in\rowspace(\mymatrix)$, we also have
$\vz\in\rowspace(\mymatrix)$, and therefore $\vz\in\signalcone$.
Moreover, $\vz\neq\vzero$, because otherwise
$\signalvec\primed=\scalescalar\primed\signalvec$, contradicting
$\signalvec\primed\notin\ray(\signalvec)$.
Finally, by the definition of $\scalescalar\primed$, at least one coordinate of $\vz$ is zero,
while $\vz$ is supported inside $\supp(\signalvec)$. Thus, $\supp(\vz)\subsetneq \supp(\signalvec)$.
This contradicts condition (ii). Therefore
$\signalvec\primed\in\ray(\signalvec)$.

The same argument applies to $\signalvec\doubleprimed$, so
$\signalvec\doubleprimed\in\ray(\signalvec)$ as well. 
Thus, $\ray(\signalvec)$ is an extreme ray of $\signalcone$.
\end{proof}

\begin{proof}[Proof of \Cref{lem:extreme-bag-undominated}]
Fix a predictor $\classifiernew$, let $\signalbag(\classifiernew)=\{\signalvec_{\classifiernew,\ell}\}_{\ell\in[\cellnum_\classifiernew]}$ be its reduced bag-of-vectors.
Assume that predictor $\classifiernew\partitiondomeq \classifier$.
By \Cref{lem:partition dom bag equiv}, there exists a row-stochastic matrix
$\kernalmtx=(\kernalmtx_{\ell,r})_{\ell\in[\cellnum_\classifiernew],\,r\in[\cellnum_{\classifier}]}$ such that signal vector satisfies
\begin{align}
    \signalvec_{\classifier,r}
    =
    \sum\nolimits_{\ell\in[\cellnum_\classifiernew]} \kernalmtx_{\ell,r}\,\signalvec_{\classifiernew,\ell}
    \quad
    \text{for every } r\in[\cellnum_{\classifier}]~.
    \label{eq:extreme-bag-dom}
\end{align}
Fix a target signal index $r\in[\cellnum_{\classifier}]$.
Because signal vector $\signalvec_{\classifier,r}$ spans an extreme ray of constructible signal cone $\signalcone$ and each signal vector
$\signalvec_{\classifiernew,\ell}$ belongs to $\signalcone$, every source signal index $\ell$ with $\kernalmtx_{\ell,r}>0$ must satisfy $\signalvec_{\classifiernew,\ell}\in \ray(\signalvec_{\classifier,r})$.
Since $\signalbag(\classifiernew)$ is reduced, there is at most one such $\ell$.
Denote it by $\ell_r$ if it exists.
Then Eqn.~\eqref{eq:extreme-bag-dom} reduces to
\begin{align*}
    \signalvec_{\classifier,r}
    =
    \kernalmtx_{\ell_r,r}\,\signalvec_{\classifiernew,\ell_r}~.
\end{align*}
Thus, the vector $\signalvec_{\classifiernew,\ell_r}$ lies on the same ray as $\signalvec_{\classifier,r}$.

We claim that row $\ell_r$ of $\kernalmtx$ has no positive entry in any other column $r'\neq r$.
If $\kernalmtx_{\ell_r,r'}>0$ for some $r'\neq r$, then the same argument applied to column $r'$
shows that $\signalvec_{\classifiernew,\ell_r}$ lies on the same ray as $\signalvec_{\classifier,r'}$.
Therefore $\signalvec_{\classifier,r}$ and $\signalvec_{\classifier,r'}$ would lie on the same positive ray,
contradicting the fact that $\signalbag(\classifier)$ is reduced.
Thus the only positive entry in row $\ell_r$ is the one in column $r$.
Since every row of $\kernalmtx$ sums to $1$, it follows that
$\kernalmtx_{\ell_r,r}=1$, and thus, we have 
$\signalvec_{\classifiernew,\ell_r}
=
\signalvec_{\classifier,r}$.

We have shown that every vector in $\signalbag(\classifier)$ appears as a distinct row of $\signalbag(\classifiernew)$.
If $\signalbag(\classifiernew)$ contained any additional row $\signalvec_{\classifiernew,\ell}$ not equal to one of these rows,
then because row $\ell$ of $\kernalmtx$ has sum $1$, it would have to satisfy
$\kernalmtx_{\ell,r}>0$ for some $r$.
By the first paragraph of the proof, $\signalvec_{\classifiernew,\ell}$ would then lie on the same ray as
$\signalvec_{\classifier,r}$, contradicting the reducedness of $\signalbag(\classifiernew)$.
Thus $\signalbag(\classifiernew)=\signalbag(\classifier)$ up to reordering.

In particular, if some predictor $\classifiernew\in \classifierfam$ satisfied $\classifiernew\partitiondom \classifier$, then
$\classifiernew\partitiondomeq \classifier$ would imply $\signalbag(\classifiernew)=\signalbag(\classifier)$.
Applying \Cref{lem:partition dom bag equiv} again, this yields
$\classifier\partitiondomeq 
\classifiernew$, contradicting $\classifiernew\partitiondom \classifier$.
Thus, no constructible predictor strictly dominates $\classifier$, so $\classifier\in \partitiondompare$.
\end{proof}

\section{Signal-Level Calibration and Prediction-Level Calibration}
\label{apx:calibration}
\aiedit{
In this section, we discuss signal-level calibration, adopted throughout the paper, and conventional prediction-level calibration. 
We explain how prediction-level calibrated inputs can be represented within our framework and clarify the implications for informational refinement.

\xhdr{Signal-level versus prediction-level calibration}
Recall that a predictor $\givenclassifier=(\signalspace_\givenclassifier,\kernal_\givenclassifier,\prediction_\givenclassifier)$ generates a signal $\signalrea_\givenclassifier$ and assigns the numerical prediction $\predictionrea_\givenclassifier=\prediction_\givenclassifier(\signalrea_\givenclassifier)$. Signal-level calibration requires
\begin{align*}
    \expect{\labelrea\mid\signalrea_\givenclassifier=\signal}
    &=\prediction_\givenclassifier(\signal)
    \quad\text{for every }\signal\in\signalspace_\givenclassifier
    \text{ with }\prob{\signalrea_\givenclassifier=\signal}>0~.
\end{align*}
In contrast, prediction-level calibration requires only
\begin{align*}
    \expect{\labelrea\mid\predictionrea_\givenclassifier=\predictionvalue}
    &=\predictionvalue
    \quad\text{for every }\predictionvalue\in[0,1]
    \text{ with }\prob{\predictionrea_\givenclassifier=\predictionvalue}>0~.
\end{align*}
As shown in \Cref{sec:prelim}, signal-level calibration implies prediction-level calibration, but the converse need not hold when distinct signals share the same numerical prediction. Prediction-level calibration constrains the average label within each group of signals having the same numerical prediction, rather than the conditional label mean of each signal individually. When the prediction map is injective on the positive-probability signals, the two notions coincide, since observing the numerical prediction then determines the signal. We do not impose injectivity: our framework deliberately allows distinct signals with identical numerical predictions.

\xhdr{Representing prediction-level calibrated inputs}
Consider a predictor $\givenclassifier=(\signalspace_\givenclassifier,\kernal_\givenclassifier,\prediction_\givenclassifier)$ that satisfies prediction-level calibration. We define a new predictor $\coarsenedclassifier$ by merging signals that yield the same numerical prediction. Its signal space is $\signalspace_{\coarsenedclassifier}=\prediction_\givenclassifier(\signalspace_\givenclassifier)$, the set of distinct numerical predictions generated by $\givenclassifier$. For every $\predictionvalue\in\signalspace_{\coarsenedclassifier}$ and $\fea\in\feaspace$, define
\begin{align*}
    \kernal_{\coarsenedclassifier}(\predictionvalue\mid\fea)
    &\triangleq
    \sum\nolimits_{\signal\in\signalspace_\givenclassifier:
                    \prediction_\givenclassifier(\signal)=\predictionvalue}
    \kernal_\givenclassifier(\signal\mid\fea)~,
    \quad
    \prediction_{\coarsenedclassifier}(\predictionvalue)\triangleq\predictionvalue~.
\end{align*}
By construction, the event that $\coarsenedclassifier$ generates signal $\predictionvalue$ coincides with the event that $\givenclassifier$ outputs numerical prediction $\predictionvalue$. Thus, for every such signal of positive probability, prediction-level calibration of $\givenclassifier$ implies
\begin{align*}
    \expect{\labelrea\mid\signalrea_{\coarsenedclassifier}=\predictionvalue}
    &=\expect{\labelrea\mid\predictionrea_\givenclassifier=\predictionvalue}
    =\predictionvalue
    =\prediction_{\coarsenedclassifier}(\predictionvalue)~.
\end{align*}
Consequently, $\coarsenedclassifier$ satisfies signal-level calibration and preserves the conditional distribution of numerical predictions given each feature.

This transformation incorporates prediction-level calibrated inputs into our signal-based framework, provided that the resulting signal kernels are available. The observable system is then built from the coarsened signal vectors, for which prediction-level calibration of the original inputs supplies the required label masses. The transformation may lose information carried by the original signals, but those distinctions are unavailable when only numerical predictions are observed. In particular, coarsening by numerical value differs from ray reduction: equal-valued signals need not lie on the same positive ray.

\xhdr{Implications for informational refinement}
The input transformation does not make signal-based informational refinement and refinement based only on numerical predictions equivalent. Two output signals may carry different information about the feature while receiving the same calibrated numerical prediction. They remain distinguishable under our signal-based formulation, but become indistinguishable when only their numerical predictions are retained.

To illustrate this distinction, consider three equally likely features A, B, and C, with Bayes probabilities $\bayesprob_{\text{A}}=\bayesprob_{\text{B}}=0.2$ and $\bayesprob_{\text{C}}=0.8$. Suppose $\givenclassifier_1$ has signal partition $\{\{\text{A},\text{B}\},\{\text{C}\}\}$, while $\givenclassifier_2$ has signal partition $\{\{\text{A}\},\{\text{B},\text{C}\}\}$. Their respective calibrated numerical predictions are $(0.2,0.8)$ and $(0.2,0.5)$. The input signal vectors span $\reals^3$: they include $(1,0,0)^\top$ and $(0,0,1)^\top$, and
\begin{align*}
    (0,1,0)^\top=(1,1,0)^\top-(1,0,0)^\top~.
\end{align*}
Hence, the fully separating predictor $\classifier$, whose signals distinguish A, B, and C, is constructible and strictly informationally dominates both inputs under \Cref{defn:partition dom}.

However, the numerical predictions of $\classifier$ are $(0.2,0.2,0.8)$. Merging its first two signals because they produce the same numerical prediction gives a predictor that cannot distinguish A from B. Any feature-independent post-processing of this merged predictor must have the same conditional output distribution on A and B, whereas $\givenclassifier_2$ predicts $0.2$ on A and $0.5$ on B. Therefore, the merged predictor cannot even simulate the numerical predictions of $\givenclassifier_2$.

Thus, merging equal-valued output signals does not necessarily preserve informational dominance over the original inputs. The maximality and refinement guarantees established in this paper concern signal-based informational dominance and do not automatically extend to a formulation in which only numerical predictions are observable outputs. In particular, input coarsening alone does not justify transferring these guarantees through numerical-output post-processing.

Our choice of signal-level calibration and signal-based informational dominance therefore defines a distinct information-refinement problem. It preserves informational distinctions among signals, supports the observable linear construction in \Cref{subsec:constructible and linear-infor}, and is the formulation addressed by the algorithmic results in \Cref{sec:det,sec:random}.

}

\end{document}